\pdfoutput=1
\RequirePackage{ifpdf}
\ifpdf % We are running pdfTeX in pdf mode
\documentclass[pdftex]{sigma}
\else
\documentclass{sigma}
\fi

\newcommand{\so}{\scriptscriptstyle \rm I}
\newcommand{\st}{\scriptscriptstyle \rm I\hspace{-1pt}I}
\newcommand{\diag}{\operatorname{diag}}
\newcommand{\tr}{\operatorname{tr}}
\newcommand{\Res}{\operatorname{Res}}
\newcommand{\bbb}{\mathbf{B}}
\newcommand{\ccc}{\mathbf{C}}
\newcommand{\bu}{\bar u}
\newcommand{\bv}{\bar v}
\newcommand{\bw}{\bar w}

\newcommand{\bth}{\bar{\theta}}
\newcommand{\vk}{\kappa}
\newcommand{\tvk}{\widetilde{\kappa}}

\numberwithin{equation}{section}
\newtheorem{Theorem}{Theorem}[section]
\newtheorem{Lemma}[Theorem]{Lemma}
\newtheorem{Proposition}[Theorem]{Proposition}
\newtheorem{Corollary}[Theorem]{Corollary}
{ \theoremstyle{definition}
\newtheorem{Remark}[Theorem]{Remark}
\newtheorem{Definition}[Theorem]{Definition}}

\begin{document}
\allowdisplaybreaks

\newcommand{\arXivNumber}{1906.06897}

\renewcommand{\PaperNumber}{066}

\FirstPageHeading

\ShortArticleName{Scalar products in twisted XXX spin chain. Determinant Representation}

\ArticleName{Scalar Products in Twisted XXX Spin Chain.\\ Determinant Representation}

\Author{Samuel BELLIARD~$^\dag$ and Nikita A.~SLAVNOV~$^\ddag$}

\AuthorNameForHeading{S.~Belliard and N.A.~Slavnov}

\Address{$^\dag$~Institut Denis-Poisson, Universit\'e de Tours, Universit\'e d'Orl\'eans,\\
\hphantom{$^\dag$}~Parc de Grammont, 37200 Tours, France}
\EmailD{\href{mailto:email@address}{samuel.belliard@gmail.com}}

\Address{$^\ddag$~Steklov Mathematical Institute of Russian Academy of Sciences,\\
\hphantom{$^\ddag$}~8 Gubkina Str., Moscow, 119991, Russia}
\EmailD{\href{mailto:email@address}{nslavnov@mi-ras.ru}}

\ArticleDates{Received June 19, 2019, in final form August 27, 2019; Published online September 03, 2019}

\Abstract{We consider XXX spin-$1/2$ Heisenberg chain with non-diagonal boundary conditions. We obtain a compact determinant representation for the scalar product of on-shell and off-shell Bethe vectors. In the particular case when both Bethe vectors are on shell, we obtain a determinant representation for the norm of on-shell Bethe vector and prove orthogonality of the on-shell vectors corresponding to the different eigenvalues of the transfer matrix.}


\Classification{82B23; 81R50}

\section{Introduction}

The algebraic Bethe ansatz (ABA) \cite{FadLH96, FadST79,FadT79} is a powerful method to study quantum integrable systems. Besides the studying the spectra of quantum Hamiltonians, this method is also used to calculate the correlation functions \cite{GohKS04,KitKMST12,KitMT00, BogIK93L}. The main tool for solving this problem within the framework of the ABA is the calculation of scalar products of Bethe vectors. In this context, it should be noted works \cite{Gau72,Gaud83,Kor82} in which the norm of the Hamiltonian eigenstate (on-shell Bethe vector) was computed, and also paper~\cite{Sla89}, where a compact determinant formula was obtained for the scalar product of off-shell and on-shell Bethe vectors (OFS-ONS scalar product).

The results listed above concern models possessing $U(1)$ symmetry. However, in real physical systems this symmetry can often be violated, for example, because of non-trivial boundary conditions. The study of quantum integrable models without $U(1)$ symmetry has led to the development of new techniques to perform the ABA~\cite{Skl88}.
Among the proposed methods such as the {\it off-diagonal Bethe ansatz} \cite{CYSW13a,CYSW13b,singsol,CYSW13c} or the {\it separation of the variables} \cite{Der,KitMNT17, Skl85}, the {\it modified algebraic Bethe ansatz} (MABA) allows one to understand this new method from the ABA point of view \cite{ABGP15,Bel15, BelC13,BP152,Cram14}.

This paper is a continuation of the series of works \cite{BelP,BSV18,BelSV18c} devoted to the study of the MABA. We consider closed XXX spin-$\frac{1}{2}$ chain with the Hamiltonian
\begin{gather}\label{HXXX}
H= \sum_{k=1}^N\big(\sigma^{x}_k\otimes\sigma^{x}_{k+1}+\sigma^{y}_k\otimes\sigma^{y}_{k+1}+\sigma^{z}_k\otimes\sigma^{z}_{k+1}\big),
\end{gather}
subject to the following non-diagonal boundary conditions
\begin{gather}
 \label{cs1} \gamma \sigma^{x}_{N+1}= \frac{ \tvk^2+ \vk^2-\vk_+^2-\vk_-^2}{2}\sigma^{x}_1 +i\frac{\vk^2-\tvk^2-\vk_+^2+\vk_-^2}{2}\sigma^{y}_1 +(\vk \vk_- -\tvk \vk_+)\sigma^{z}_1, \\
 \label{cs2} \gamma \sigma^{y}_{N+1}= i\frac{\tvk^2- \vk^2-\vk_+^2+\vk_-^2}{2}\sigma^{x}_1 +\frac{\tvk^2+\vk^2+\vk_+^2+\vk_-^2}{2}\sigma^{y}_1 -i(\tilde \vk \vk_++\vk \vk_-)\sigma^{z}_1, \\
 \label{cs3} \gamma \sigma^{z}_{N+1}=(\vk \vk_+-\tilde \vk \vk_-)\sigma^{x}_1 +i(\tilde \vk \vk_-+\vk \vk_+)\sigma^{y}_1 +(\tilde \vk \vk+\vk_+ \vk_-)\sigma^{z}_1.
\end{gather}
The twist parameters $\{\vk,\tvk,\vk_+,\vk_-\}$ are generic complex numbers and $\gamma =\tvk\vk -\vk_+\vk_-$. The Pauli matrices\footnote{$\sigma^{z}=\left(\begin{smallmatrix}
 1 & 0\\
 0 & -1 \end{smallmatrix}\right)$, $\sigma^{+}=\left(\begin{smallmatrix}
 0 & 1\\
 0 & 0 \end{smallmatrix}\right)$, $\sigma^{-}=\left(\begin{smallmatrix}
 0 & 0\\
 1 & 0 \end{smallmatrix}\right)$, $\sigma^{x}=\sigma^{+}+\sigma^{-}$, $\sigma^{y}={\rm i}(\sigma^{-}-\sigma^{+})$.} $\sigma^{\alpha}_k$ with $\alpha=x,y,z$ act non-trivially on the $k$th component of the quantum space $\mathcal{H}= \otimes_{k=1}^N V_k$ with $V_k=\mathbb{C}^2$.

The main object of our study is the OFS-ONS scalar product. It was conjectured in~\cite{BelP} that this scalar product admits a determinant representation similar to the one obtained in~\cite{KitMT99, Sla89} in the case of the usual ABA. In this paper we prove this conjecture. The main tool of our prove is an analog of Izergin--Korepin formula for the scalar product of off-shell Bethe vectors~\cite{Ize87, Kor82} (see~\eqref{SP-finuv}). This formula was generalized for the case of the MABA in~\cite{BelSV18c}. In this paper we specify it to the particular case when one of the Bethe vectors is on-shell. This allows us to compute the sum over partitions of the Bethe parameters in the form of a single determinant.

In the particular case, the obtained determinant representation describes the norm of on-shell Bethe vector. This representation also allows us to prove orthogonality of the on-shell vectors corresponding to the different eigenvalues of the transfer matrix.

The paper is organised as follows. In Section~\ref{MA-IK} we recall basic notions of the MABA and define modified Bethe vectors. We also introduce notation used in the paper. In Section~\ref{A-MDID} we give definition of a modified Izergin determinant, which is one of the most important tools for studying scalar products of the modified Bethe vectors. Section~\ref{S-MR} contains the main results of the paper. Here we give different determinant representations for OFS-ONS scalar product, a~determinant formula for the norm of the modified on-shell Bethe vector, and different forms of the inhomogeneous Bethe equations. In the rest of the paper we present the proofs of the results of Section~\ref{S-MR}. In Section~\ref{S-AFBE} we prove alternative forms of the inhomogeneous Bethe equations. In Section~\ref{S-MIDOS} we consider some properties of the on-shell modified Izergin determinant. Finally, in Section~\ref{S-DPSP} we present the proof of determinant formula for OFS-ONS scalar product. In particular, we prove the equivalence of two determinant representations for OFS-ONS scalar product in Section~\ref{S-TJ}. Several useful formulas and auxiliary lemmas are gathered in appendices. Appendix~\ref{A-PMID} contains a list of properties of the modified Izergin determinant. In Appendix~\ref{A-IRF}, we give some identities for rational functions. Appendices~\ref{A-ROSID} and~\ref{A-PPL} are devoted to the properties of the on-shell modified Izergin determinants.

\section{Basic notions}\label{MA-IK}

To describe the Hamiltonian \eqref{HXXX} within the framework of the quantum inverse scattering method (QISM), we first
introduce a $\mathfrak{gl}_2$-invariant $R$-matrix acting
in $\mathbb{C}^2\otimes \mathbb{C}^2$:
\begin{gather}\label{Rm}
 R(u,v)=\frac{u-v}c\mathbb{I}+P.
\end{gather}
Here $c$ is a constant, $\mathbb{I}$ is the identity operator, and $P$ is the permutation operator. The $R$-matrix~\eqref{Rm} is called~$\mathfrak{gl}_2$-invariant, because
\begin{gather*}%\label{twist-inv}
[R(u,v),K\otimes K]=0,
\end{gather*}
for any matrix $K\in \mathfrak{gl}_2$.

The key object of the QISM is a quantum monodromy matrix
\begin{gather}\label{MonoT}
T(u)=\begin{pmatrix} t_{11}(u)& t_{12}(u)\\ t_{21}(u)&t_{22}(u)\end{pmatrix}.
\end{gather}
This matrix acts as a $2\times 2$ matrix in auxiliary space $\mathbb{C}^2$. The entries $t_{ij}(u)$ are operators depending on a complex parameter $u$ and acting in the Hilbert space $\mathcal{H}$ of the Hamiltonian~\eqref{HXXX}. The commutation relations of these operators are given by an $RTT$-relation
\begin{gather*}%\label{RTT}
 R(u,v)\bigl(T(u)\otimes I \bigr) \bigl(I \otimes T(v) \bigr)=\bigl(I \otimes T(v) \bigr) \bigl(T(u)\otimes I \bigr)R(u,v).
\end{gather*}
Here $I$ is the identity operator in $\mathbb{C}^2$. Equivalently, these commutation relations can be written in the form
\begin{gather*}%\label{genCR}
[t_{ij}(u),t_{kl}(v)]=g(u,v)\bigl(t_{kj}(v)t_{il}(u)-t_{kj}(u)t_{il}(v)\bigr).
\end{gather*}

To obtain the Hamiltonian \eqref{HXXX} we first introduce an inhomogeneous monodromy matrix
\begin{gather}\label{MonodR}
T(u)=R_{0N}(u,\theta_N)\cdots R_{01}(u,\theta_1),
\end{gather}
where $\theta_j$ are inhomogeneity parameters. Each $R$-matrix $R_{0k}(u,\theta_k)$ in \eqref{MonodR} acts non-trivially in the space
$V_0\otimes V_k$, where $V_0$ is the auxiliary space of the monodromy matrix and $V_k$ is the quantum space associated with the $k$th
site of the chain. The Hamiltonian of the XXX chain with periodic boundary conditions can be obtained from the monodromy
matrix \eqref{MonodR} in the homogeneous limit $\theta_j=0$, $j=1,\dots,N$ (see \cite{Bax72}). In order to obtain the Hamiltonian
\eqref{HXXX} with the boundary condition \eqref{cs1}--\eqref{cs3} we introduce a twisted monodromy matrix
\begin{gather*}%\label{TwMM}
T_K(u)=KT(u), \qquad K=\begin{pmatrix}\tvk&\vk_+\\\vk_-&\vk\end{pmatrix}.
\end{gather*}
Then, defining a twisted transfer matrix $\mathcal{T}(u)$ by
\begin{gather*}%\label{TwTM}
\mathcal{T}(u)=\tr T_K(u),
\end{gather*}
we obtain
\begin{gather*}\label{Hfromt}
H=2c \frac{{\rm d}}{{\rm d} u}\big(\log\big(\mathcal{T}(u)\big)\big)\big|_{u\to 0, \, \theta_i\to 0}-N.
\end{gather*}

It is convenient to present the twist matrix $K$ in the form
\begin{gather*}%\label{KBDA}
K=BDA,
\end{gather*}
where
\begin{gather*}%\label{Mat-Tf}
A=\sqrt\mu\begin{pmatrix}1&\dfrac{\rho_2}{\vk^-}\\ \dfrac{\rho_1}{\vk^+}&1\end{pmatrix},\qquad
B=\sqrt\mu\begin{pmatrix} 1&\dfrac{\rho_1}{\vk^-}\\ \dfrac{\rho_2}{\vk^+}&1\end{pmatrix},\qquad
D=\begin{pmatrix} \tvk-\rho_1&0\\0& \vk-\rho_2\end{pmatrix},
\end{gather*}
and the parameters $\rho_i$ and $\mu$ enjoy the following constraints:
\begin{gather}\label{murho}
\rho_1\rho_2-\rho_2\tvk-\rho_1\vk+\vk^+\vk^-=0, \qquad \mu=\frac{1}{1-\frac{\rho_1\rho_2}{\vk^+\vk^-}}.
\end{gather}
Then we have
\begin{gather}\label{trTK}
\mathcal{T}(u)=\tr \bigl(D\overline{T}(u)\bigr), \qquad \overline{T}(u)=AT(u)B=
\begin{pmatrix} \nu_{11}(u)& \nu_{12}(u)\\ \nu_{21}(u)&\nu_{22}(u)\end{pmatrix}.
\end{gather}
Thus, instead of the twisted monodromy matrix $T_K(u)$ we can consider a modified monodromy matrix
$\overline{T}(u)$ \eqref{trTK}. Respectively, the transfer matrix $\mathcal{T}(u)$ now can be understood
as the trace of the twisted modified monodromy matrix
$\overline{T}(u)$ with the diagonal twist $D$.

\subsection{Highest weight representation of the Yangian}\label{SS-HWR}

Recall that the Hilbert space of the Hamiltonian \eqref{HXXX} is $\mathcal{H}= \otimes_{k=1}^N V_k$, where $V_k=\mathbb{C}^2$. We define a highest weight vector $|0\rangle\in \mathcal{H}$ as the state with all spins up
\begin{gather*}%\label{vacdef}
|0\rangle=\left(\begin{smallmatrix}1\\0\end{smallmatrix}\right)_1\otimes\dots \otimes
\left(\begin{smallmatrix}1\\0\end{smallmatrix}\right)_N.
\end{gather*}
Then the action of the monodromy matrix entries $t_{ij}(u)$ \eqref{MonoT} on $|0\rangle$ is
\begin{gather}
t_{ii}(u)|0\rangle=\lambda_i(u)|0\rangle,\qquad i=1,2,\nonumber\\
 t_{21}(u)|0\rangle=0,\label{HWRG}
\end{gather}
where
\begin{gather}\label{lambdas}
\lambda_1(u)=\frac1{c^N}\prod_{k=1}^N(u-\theta_k+c),\qquad \lambda_2(u)=\frac1{c^N}\prod_{k=1}^N(u-\theta_k).
\end{gather}
The representation space of the Yangian is then spanned by the Bethe vectors $\bbb^m_0(\bv)$
\begin{gather*}%\label{BVdef}
\bbb^m_0(\bv)= \prod_{i=1}^m t_{12}(v_i)|0\rangle,\qquad m=0,1,\dots,N,
\end{gather*}
which provide a formal basis depending on the parameters $\bv=\{v_1,\dots,v_m\}$.

To study scalar products of Bethe vectors we also consider the dual highest weight vector $\langle0|=|0\rangle^T$ belonging to the dual space $\mathcal{H}^*$. This vector possesses the following properties
\begin{gather*}
\langle0| t_{ii}(u)=\lambda_i(u)\langle0|, \qquad i=1,2, \qquad
\langle0| t_{12}(u)=0, \qquad
\langle0| 0\rangle=1,%\label{dHWRG}
\end{gather*}
where the functions $\lambda_i(u)$ are given by \eqref{lambdas}. Dual Bethe vectors are constructed by successive application of the operator $t_{21}$ to the vector $\langle0|$
\begin{gather*}%\label{dBVdef}
\ccc^m_0(\bv)= \langle0|\prod_{i=1}^m t_{21}(v_i),\qquad m=0,1,\dots,N.
\end{gather*}

Within the framework of the MABA the states of the space $\mathcal{H}$ (resp.\ the dual space $\mathcal{H}^*$) are generated by the successive application of the operators $\nu_{12}$ (resp.~$\nu_{21}$) to the state $|0\rangle$ (resp.~$\langle0|$). The modified Bethe vectors are given by
\begin{gather}\label{modvect}
\bbb^m(\bv)= \prod_{i=1}^m \nu_{12}(v_i)|0\rangle,\qquad \ccc^m(\bv)= \langle0|\prod_{i=1}^m \nu_{21}(v_i).
\end{gather}
In these formulas, $\bv=\{v_1,\dots,v_m\}$ are generic complex numbers. They are called {\it Bethe para\-me\-ters}.

Observe that despite of the new operators $\nu_{ij}$ satisfy the same commutation relations as the~$t_{ij}(z)$,
\begin{gather}\label{genCRn}
[\nu_{ij}(u),\nu_{kl}(v)]=g(u,v)\bigl(\nu_{kj}(v)\nu_{il}(u)-\nu_{kj}(u)\nu_{il}(v)\bigr),
\end{gather}
their \looseness=-1 actions on the highest weight vector \eqref{HWRG} change. It is easy to see that now they are given by{\samepage
\begin{gather*}
\nu_{11}(u)|0\rangle= \lambda_1(u)|0\rangle+\beta_2\nu_{12}(u)|0\rangle, \nonumber\\
\nu_{22}(u)|0\rangle= \lambda_2(u)|0\rangle+\beta_1\nu_{12}(u)|0\rangle,\nonumber\\
\nu_{21}(u)|0\rangle=\big(\beta_1\lambda_1(u)+\beta_2\lambda_2(u)\big)|0\rangle+\beta_1\beta_2\nu_{12}(u)|0\rangle,%\label{act-sing}
\end{gather*}
where $\beta_i=\rho_i/\vk^+$.}

Concluding this section we would like to mention that the operators $t_{ii}(u)$ are polynomials in $u$ of degree $N$, while $t_{ij}(u)$ with $i\ne j$ are polynomials in $u$ of degree $N-1$. This statement follows from representation~\eqref{MonodR} and~\eqref{Rm}. Since for generic twist parameters, any $\nu_{ij}(u)$ is a~linear combination of all $t_{kl}(u)$, we conclude that the operators $\nu_{ij}(u)$ are polynomials in $u$ of degree $N$ for all $i$ and $j$.

\subsection{Shorthand notation}

Before moving on, we introduce a notation and several important conventions that will be used throughout the paper. First of all, we introduce the following rational functions:
\begin{gather}
g(u,v)=\frac{c}{u-v},\qquad f(u,v)=1+g(u,v)=\frac{u-v+c}{u-v}, \nonumber\\
h(u,v)= \frac{f(u,v)}{g(u,v)}= \frac{u-v+c}{c}.\label{gfh}
\end{gather}
Actually, all these functions depend on the difference of their arguments. However, we do not stress this dependence, which will allow us to introduce special shorthand notation for their products. It is easy to see that the functions introduced above possess the following properties:
\begin{gather*}%\label{x-gfh}
\chi(u,v)\Bigr|_{c\to -c}=\chi(v,u), \qquad \chi(-u,-v)=\chi(v,u),\qquad \chi(u-c,v)=\chi(u,v+c),
\end{gather*}
where $\chi$ is any of the three functions. One can also convince himself that
\begin{gather}\label{gfh-prop}
g(u,v-c)=\frac{1}{h(u,v)},\qquad h(u,v+c)= \frac{1}{g(u,v)}, \qquad f(u,v+c)=\frac{1}{f(v,u)}.
\end{gather}

Let us formulate now a convention on the notation. We denote sets of variables by a bar, for example, $\bu=\{u_1,\dots,u_n\}$. Notation $\bu\pm c$ means that $\pm c$ is added to all the arguments of the set~$\bu$. Individual elements of the sets or subsets are denoted by Latin subscripts, for instance, $u_j$ is an element of $\bu$. As a rule, the number of elements in the sets is not shown explicitly in the equations, however we give these cardinalities in special comments to the formulas, if necessary.

We also consider subsets of variables. We agree upon that the notation $\bu_k$ refers to a subset that is complementary to the element~$u_k$, that is, $\bu_k=\bu\setminus u_k$. In all other cases, we denote subsets by subscripts so that they can be easily distinguished from the elements of sets. In particular, when dealing with partitions of the sets into subsets we mostly denote the latter by the Roman numbers $\bu_{\so}$, $\bu_{\st}$, and so on. Notation $\{\bu_{\so},\bu_{\st}\}\vdash\bu$ means that the set $\bu$ is divided into two disjoint subsets $\bu_{\so}$ and $\bu_{\st}$. The order of the elements in each subset is not essential.

To make the formulas more compact, we use a shorthand notation for the products of the rational functions \eqref{gfh}, the operators~$\nu_{kl}(u)$ \eqref{MonoT}, and the vacuum eigenvalues~$\lambda_i(u)$ \eqref{HWRG}. Namely, if a function (an operator) depends on a~(sub)set of variables, then one should take a~product with respect to the corresponding (sub)set. For example,
\begin{gather}\label{shn}
\nu_{kl}(\bu)=\prod_{u_j\in \bu} \nu_{kl}(u_j), \qquad f(z,\bu_{\so})=\prod_{u_j\in \bu_{\so}} f(z,u_j), \qquad f(\bu_k,u_k)=
\prod_{\substack{u_j\in \bu\\u_j\ne u_k}} f(u_j,u_k).
\end{gather}
Note that due to commutativity of the $\nu_{kl}$-operators (which follows from~\eqref{genCRn}) the first product in~\eqref{shn} is well defined. Notation $f(\bu,\bv)$ means a double product over the sets $\bu$ and $\bv$. By definition any product over the empty set is equal to~$1$. A double product is equal to~$1$ if at least one of the sets is empty.

In particular, using this convention we can write down the modified Bethe vectors and their dual ones \eqref{modvect} in the form
\begin{gather*}%\label{modvect1}
\bbb^m(\bv)= \nu_{12}(\bv)|0\rangle,\qquad \ccc^m(\bv)= \langle0| \nu_{21}(\bv),
\end{gather*}
where $\bv=\{v_1,\dots,v_m\}$. The eigenvalues \eqref{lambdas} take the form
\begin{gather}\label{lambdas1}
\lambda_1(u)=h(u,\bth),\qquad \lambda_2(u)=g(u,\bth)^{-1},\qquad \frac{\lambda_1(u)}{\lambda_2(u)}=f(u,\bth),
\end{gather}
where $\bth=\{\theta_1,\dots,\theta_N\}$.

\subsection{On-shell Bethe vectors}\label{S-OSBV}

Within the framework of the MABA the operator \eqref{trTK}
\begin{gather*}%\label{TrT}
\mathcal{T}(z)=\tr \bigl(D\overline{T}(z)\bigr) = (\tvk-\rho_1)\nu_{11}(z)+(\vk-\rho_2)\nu_{22}(z)
\end{gather*}
appears to be a generating function of the integrals of motion. Thus, the eigenstates of this operator also are the eigenstates of the Hamiltonian \eqref{HXXX}. They are commonly called {modified on-shell Bethe vectors}.

The on-shell Bethe vectors in the MABA solvable models were found in \cite{BelP, CYSW13a,Cram14}. Here we briefly recall this construction.

Let $\bu=\{u_1,\dots,u_N\}$, where $N$ is the number of sites of the chain. Then a modified Bethe vector $\bbb^N(\bu)$ \eqref{modvect} becomes on-shell, if the Bethe parameters $\bu$ satisfy a system of modified Bethe equations
\begin{gather}
(\vk-\rho_2)\lambda_2(u_j)f(u_j,\bu_j)-(\tvk-\rho_1)\lambda_1(u_j)f(\bu_j,u_j)\nonumber\\
\qquad{} +(\rho_1+\rho_2) g(u_j,\bu_j)\lambda_1(u_j)\lambda_2(u_j)=0,\label{BE00}
\end{gather}
for $j=1,\dots,N$. Then
\begin{gather*}%\label{EigVeq}
\mathcal{T}(z)\bbb^N(\bu)=\Lambda(z|\bu)\bbb^N(\bu),
\end{gather*}
where the eigenvalue $\Lambda(z|\bu)$ is
\begin{gather}\label{Lam}
\Lambda(z|\bu)=(\tvk-\rho_1)\lambda_1(z)f(\bu,z) +(\vk-\rho_2)\lambda_2(z)f(z,\bu) +(\rho_1+\rho_2)\lambda_1(z)\lambda_2(z)g(z,\bu).
\end{gather}

For further application, it is convenient to introduce a function (cf.~\cite{KitMST05})
\begin{gather}\label{Y-def}
\mathcal{Y}(z|\bu)=\frac{\Lambda(z|\bu)}{\lambda_2(z)g(z,\bu)}.
\end{gather}
The explicit expression for this function reads
\begin{gather*}%\label{Y-expl}
\mathcal{Y}(z|\bu)=(-1)^N(\tvk-\rho_1)\frac{\lambda_1(z)}{\lambda_2(z)}h(\bu,z) +(\vk-\rho_2)h(z,\bu) +(\rho_1+\rho_2)\lambda_1(z).
\end{gather*}
Then Bethe equations \eqref{BE00} take the form
\begin{gather*}%\label{BE-Y}
\mathcal{Y}(u_j|\bu)=0,\qquad j=1,\dots,N.
\end{gather*}

Similarly, a modified dual Bethe vector $\ccc^N(\bu)$ \eqref{modvect} becomes on-shell, if the Bethe para\-meters $\bu$ satisfy the system~\eqref{BE00}. Then
\begin{gather*}%\label{dEigVeq}
\ccc^N(\bu) \mathcal{T}(z)=\Lambda(z|\bu)\ccc^N(\bu),
\end{gather*}
with the eigenvalue \eqref{Lam}.

The main goal of this paper is to study the scalar products of the modified Bethe vectors, in which at least one of the vectors is on-shell. In fact, it is enough to consider the case when the vector $\bbb^N(\bu)$ is on-shell, while the dual state $\ccc^N(\bv)$ is a generic modified dual Bethe vector (see~\eqref{symSP} below).

\section{Modified Izergin determinant and scalar products} \label{A-MDID}

Within the framework of the ABA, the Izergin determinant allows us to obtain a formula for the scalar product of Bethe vectors as a sum over partitions of the Bethe parameters~\cite{Ize87, Kor82}. For the MABA, we should introduce the modified Izergin determinant. Then one can obtain an analogous formula for the scalar product of the modified Bethe vectors~\cite{BelSV18c}.

\subsection{Modified Izergin determinant \label{MID}}

\begin{Definition}Let $\bu=\{u_1,\dots,u_n\}$, $\bv=\{v_1,\dots,v_m\}$, and $z$ be complex numbers. Then the modified Izergin determinant $K_{n,m}^{(z)}(\bu|\bv)$ is defined by
\begin{gather}\label{defKdef1}
K_{n,m}^{(z)}(\bu|\bv)=\det_m\left(-z\delta_{jk}+\frac{f(\bu,v_j)f(v_j,\bv_j)}{h(v_j,v_k)}\right).
\end{gather}
Alternatively the modified Izergin determinant can be presented as
\begin{gather}\label{defKdef2}
K_{n,m}^{(z)}(\bu|\bv)=(1-z)^{m-n}\det_n\left(\delta_{jk}f(u_j,\bv)-z\frac{f(u_j,\bu_j)}{h(u_j,u_k)}\right).
\end{gather}
\end{Definition}

Recall that we use the shorthand notation for the products \eqref{shn} in representations \eqref{defKdef1} and \eqref{defKdef2}. The proof of the equivalence of these representations can be found in~\cite{GorZZ14}. It is based on the recursive property~\eqref{resK}.

It is also convenient to introduce a conjugated modified Izergin determinant as
\begin{gather}\label{CdefKdef1}
\overline{K}_{n,m}^{(z)}(\bu|\bv)=K_{n,m}^{(z)}(\bu|\bv)\Bigr|_{c\to -c}
=\det_m\left(-z\delta_{jk}+\frac{f(v_j,\bu)f(\bv_j,v_j)}{h(v_k,v_j)}\right),
\end{gather}
or equivalently,
\begin{gather*}%\label{CdefKdef2}
\overline{K}_{n,m}^{(z)}(\bu|\bv)=(1-z)^{m-n}\det_n\left(\delta_{jk}f(\bv,u_j)-z\frac{f(\bu_j,u_j)}{h(u_k,u_j)}\right).
\end{gather*}

In the particular case $z=1$ and $\#\bu=\#\bv=n$ the modified Izergin determinant turns into the ordinary Izergin determinant, that we traditionally denote by $K_{n}(\bu|\bv)$:
\begin{gather*}%\label{ModI-OrdI}
K_{n,n}^{(1)}(\bu|\bv)=K_{n}(\bu|\bv).
\end{gather*}
This property can be seen from the recursion \eqref{resK} and the initial condition \eqref{K0}. Let us recall one more representation for the ordinary Izergin determinant~\cite{Ize87}
\begin{gather}\label{OrdIK}
K_{n}(\bu|\bv)=h(\bu,\bv)\Delta'(\bu)\Delta(\bv)\det_n \left(\frac{g(u_j,v_k)}{h(u_j,v_k)}\right),
\end{gather}
where
\begin{gather}\label{Delta}
 \Delta(\bv)=\prod_{1\le j<k\le n} g(v_k,v_j),\qquad \Delta'(\bu)=\prod_{1\le j<k\le n} g(u_j,u_k).
 \end{gather}

It is easy to see that the modified Izergin determinants $K_{n,m}^{(z)}(\bu|\bv)$ and $\overline{K}_{n,m}^{(z)}(\bu|\bv)$ are rational functions of $\bu$ and $\bv$. They are symmetric over $\bu$ and symmetric over $\bv$. Other properties of the modified Izergin determinant are collected in Appendix~\ref{A-PMID}.

\subsection{Scalar product \label{SS}}

\begin{Definition}Let $\bu=\{u_1,\dots,u_n\}$, $\bv=\{v_1,\dots,v_m\}$. The scalar product of two modified Bethe vectors is defined as
\begin{gather*}%\label{SPDEF}
S_\nu^{m,n}(\bv,\bu)=\ccc^m(\bv)\bbb^n(\bu)=\langle0| \nu_{21}(\bv)\nu_{12}(\bu)|0\rangle.
\end{gather*}
\end{Definition}

The function $S_\nu^{m,n}(\bv,\bu)$ has several important properties. First, it follows from the commutation relations $[\nu_{12}(u),\nu_{12}(v)]=0$ and $[\nu_{21}(u),\nu_{21}(v)]=0$ that the scalar product is a symmetric function of $\bv$ and a symmetric function of~$\bu$. Furthermore, since any $\nu_{ij}(z)$ is a polynomial in $z$ of degree $N$, the scalar product is a polynomial in any $v_j$ and in any $u_j$ of degree $N$.

\begin{Proposition}[\cite{BelSV18c}]\label{MSPfor}Let $\#\bu=n$ and $\#\bv=m$. Then
\begin{gather}
S_\nu^{m,n}(\bv,\bu)=\mu^{2m}(\mu-1)^{n-m}
\sum_{\substack{\{\bv_{\so},\bv_{\st}\}\vdash\bv\\\{\bu_{\so},\bu_{\st}\}\vdash\bu }}
\beta_1^{p_{\st}-q_{\st}}\beta_2^{p_{\so}-q_{\so}}\lambda_2(\bv_{\so})\lambda_1(\bv_{\st}) \lambda_2(\bu_{\st})\lambda_1(\bu_{\so})\nonumber\\
\hphantom{S_\nu^{m,n}(\bv,\bu)=}{} \times f(\bv_{\so}, \bv_{\st}) f(\bu_{\st},\bu_{\so})
 K_{q_{\st},p_{\st}}^{(1/\mu)}(\bu_{\st}|\bv_{\st}) \overline{K}_{q_{\so},p_{\so}}^{(1/\mu)}(\bu_{\so}|\bv_{\so}).\label{SP-finuv}
\end{gather}
Here $p_{\so}=\#\bv_{\so}$, $p_{\st}=\#\bv_{\st}$, $q_{\so}=\#\bu_{\so}$, and $q_{\st}=\#\bu_{\st}$. The sum is taken over all partitions $\{\bv_{\so},\bv_{\st}\}\vdash\bv$ and $\{\bu_{\so},\bu_{\st}\}\vdash\bu$. There is no restriction on the cardinalities of the subsets. The functions $K^{(1/\mu)}_{q_{\st},p_{\st}}$ and $\overline{K}^{(1/\mu)}_{q_{\so},p_{\so}}$ respectively are the modified Izergin determinants~\eqref{defKdef1} and~\eqref{CdefKdef1} at $z=1/\mu$.
\end{Proposition}

Equation \eqref{SP-finuv} was derived in \cite{BelSV18c}.

\begin{Proposition}\label{Pmor-SP}The scalar product of generic modified Bethe vectors satisfies a condition
\begin{gather}\label{symSP}
S_\nu^{m,n}(\bv,\bu)=\left(\frac{\vk_-}{\vk_+}\right)^{m-n}S_\nu^{n,m}(\bu,\bv).
\end{gather}
\end{Proposition}

\begin{proof} Replacing in \eqref{SP-finuv} $\bu\leftrightarrow\bv$, $n\leftrightarrow m$, $q_{\so}\leftrightarrow p_{\so}$, and $q_{\st}\leftrightarrow p_{\st}$ we obtain
\begin{gather}
S_\nu^{n,m}(\bu,\bv)=\mu^{2n}(\mu-1)^{m-n}
\sum_{\substack{\{\bv_{\so},\bv_{\st}\}\vdash\bv\\\{\bu_{\so},\bu_{\st}\}\vdash\bu }}
\beta_1^{q_{\st}-p_{\st}}\beta_2^{q_{\so}-p_{\so}}\lambda_2(\bu_{\so})\lambda_1(\bu_{\st}) \lambda_2(\bv_{\st})\lambda_1(\bv_{\so})\nonumber\\
\hphantom{S_\nu^{n,m}(\bu,\bv)=}{}\times f(\bu_{\so}, \bu_{\st}) f(\bv_{\st},\bv_{\so})
 K_{p_{\st},q_{\st}}^{(1/\mu)}(\bv_{\st}|\bu_{\st}) \overline{K}_{p_{\so},q_{\so}}^{(1/\mu)}(\bv_{\so}|\bu_{\so}).\label{SP-finvu}
\end{gather}
Using \eqref{c-c} we find
\begin{gather*}%\label{KK-KK}
K_{p_{\st},q_{\st}}^{(1/\mu)}(\bv_{\st}|\bu_{\st}) \overline{K}_{p_{\so},q_{\so}}^{(1/\mu)}(\bv_{\so}|\bu_{\so})
=\left(1-\frac1\mu\right)^{n-m}\overline{K}_{q_{\st},p_{\st}}^{(1/\mu)}(\bu_{\st}|\bv_{\st}) K_{q_{\so},p_{\so}}^{(1/\mu)}(\bu_{\so}|\bv_{\so}),
\end{gather*}
and hence,
\begin{gather*}
S_\nu^{n,m}(\bu,\bv)=\mu^{n+m}
\sum_{\substack{\{\bv_{\so},\bv_{\st}\}\vdash\bv\\\{\bu_{\so},\bu_{\st}\}\vdash\bu }}
\beta_1^{q_{\st}-p_{\st}}\beta_2^{q_{\so}-p_{\so}}\lambda_2(\bu_{\so})\lambda_1(\bu_{\st}) \lambda_2(\bv_{\st})\lambda_1(\bv_{\so})\nonumber\\
\hphantom{S_\nu^{n,m}(\bu,\bv)=}{}\times f(\bu_{\so}, \bu_{\st}) f(\bv_{\st},\bv_{\so})
 \overline{K}_{q_{\st},p_{\st}}^{(1/\mu)}(\bu_{\st}|\bv_{\st}) K_{q_{\so},p_{\so}}^{(1/\mu)}(\bu_{\so}|\bv_{\so}).%\label{SP-finvu1}
\end{gather*}
Relabeling the subsets as $\bu_{\so}\leftrightarrow\bu_{\st}$ and $\bv_{\so}\leftrightarrow\bv_{\st}$ (and respectively relabeling their cardinalities) we finally arrive at
\begin{gather}
S_\nu^{n,m}(\bu,\bv)=\mu^{n+m}(\beta_1\beta_2)^{n-m}
\sum_{\substack{\{\bv_{\so},\bv_{\st}\}\vdash\bv\\\{\bu_{\so},\bu_{\st}\}\vdash\bu }}
\beta_1^{p_{\st}-q_{\st}}\beta_2^{p_{\so}-q_{\so}}\lambda_2(\bu_{\st})\lambda_1(\bu_{\so}) \lambda_2(\bv_{\so})\lambda_1(\bv_{\st})\nonumber\\
\hphantom{S_\nu^{n,m}(\bu,\bv)=}{}\times f(\bu_{\st}, \bu_{\so}) f(\bv_{\so},\bv_{\st}) \overline{K}_{q_{\so},p_{\so}}^{(1/\mu)}(\bu_{\so}|\bv_{\so}) K_{q_{\st},p_{\st}}^{(1/\mu)}(\bu_{\st}|\bv_{\st}).\label{SP-finvu2}
\end{gather}
Comparing equations \eqref{SP-finvu2} and \eqref{SP-finvu} we see that
\begin{gather*}%\label{ratS}
\frac{S_\nu^{m,n}(\bv,\bu)}{S_\nu^{n,m}(\bu,\bv)}=\left(\frac{\mu\beta_1\beta_2}{\mu-1}\right)^{n-m},
\end{gather*}
and substituting here explicit expressions for $\beta_i=\rho_i/\vk_+$ and $\mu$ \eqref{murho} we immediately arrive at~\eqref{symSP}.\end{proof}

\section{Main results}\label{S-MR}

In this section we give a list of the main results obtained in this paper. Most of the proofs are given in the remaining part of the text.

\subsection{Determinant representations for the scalar product}\label{SS-DRSP}

\begin{Theorem}\label{Mtheor1} Let $\#\bu=\#\bv=N$. Let the set $\bu$ solve Bethe equations \eqref{BE00}, while the set $\bv$ consist of arbitrary complex numbers. Then the scalar product $S_\nu^N(\bv,\bu)$ of the modified on-shell Bethe vector $\bbb^N(\bu)$ and generic dual modified Bethe vector $\ccc^N(\bv)$ has the following determinant representation:
\begin{gather}\label{Svu}
S_\nu^N(\bv,\bu)=\frac{(\mu\beta)^N \Delta(\bv)\Delta'(\bu)\lambda_2(\bu)\lambda_2(\bv)}{(\beta\vk+\tvk-2\rho_1)^N}
K_{N,N}^{(-1/\beta)}(\bu|\bth) \det_N\left(\frac{c}{g(v_k,\bu)\lambda_2(v_k)}\frac{\partial\Lambda(v_k|\bu)}{\partial u_j}\right).\!\!\!
\end{gather}
Here $\Lambda(v|\bu)$ is the eigenvalue~\eqref{Lam}, $\beta=\rho_1/\rho_2$, and $ \Delta(\bv)$, $\Delta'(\bu)$ are given by \eqref{Delta} with $n=N$.
\end{Theorem}

Representation \eqref{Svu} was conjectured in~\cite{BelP}. We also would like to point out that similarly to the scalar products of on-shell and off-shell Bethe vectors in ABA \cite{KitMT99, Sla89} this representation involves Jacobian of the transfer matrix eigenvalue:
\begin{gather}
\frac{c}{g(v_k,\bu)\lambda_2(v_k)}\frac{\partial\Lambda(v_k|\bu)}{\partial u_j}
=(-1)^{N-1}(\tvk-\rho_1)\frac{\lambda_1(v_k)}{\lambda_2(v_k)}\frac{g(u_j,v_k)}{h(u_j,v_k)}h(\bu,v_k) \nonumber\\
\hphantom{\frac{c}{g(v_k,\bu)\lambda_2(v_k)}\frac{\partial\Lambda(v_k|\bu)}{\partial u_j}=}{}
+(\vk-\rho_2)\frac{g(v_k,u_j)}{h(v_k,u_j)}h(v_k,\bu) +(\rho_1+\rho_2) \lambda_1(v_k)g(v_k,u_j).\label{Lamder0}
\end{gather}

\begin{Theorem}\label{Mtheor2}Under the condition of Theorem~{\rm \ref{Mtheor1}} the scalar product has the following determinant representation:
\begin{gather}\label{Svunew1}
S_\nu^N(\bv,\bu)=\lambda_2(\bu)\lambda_2(\bv)\left(\frac{\mu}{\alpha}\right)^N
K_{N,N}^{(\mu+(\mu-1)/\beta)}(\bu|\bth)K_{2N,N}^{(\alpha)}(\{\bu,\bv\}|\bth),
\end{gather}
where
\begin{gather}\label{alpha0}
\alpha=\frac{\vk-\rho_2}{\tvk-\rho_1}.
\end{gather}
\end{Theorem}

The equivalence of representations \eqref{Svu} and \eqref{Svunew1} is proved in Section~\ref{S-TJ}.

\begin{Corollary}\label{ortho}Dual and ordinary modified on-shell Bethe vectors corresponding to different eigenvalues are orthogonal.
\end{Corollary}

\begin{proof} Let
\begin{gather*}%\label{gaj}
\gamma_j=\frac{g(u_j,\bu_j)}{g(u_j,\bv)}.
\end{gather*}
Observe that if the vectors $\bbb^N(\bu)$ and $\ccc^N(\bv)$ correspond to different transfer matrix eigenvalues (i.e., $\bu\ne\bv$), then there exists at least one $\gamma_j\ne 0$.

Using formulas of Appendix~\ref{A-SRF} we obtain
\begin{gather}
\sum_{j=1}^N \frac{c\gamma_j}{g(v_k,\bu)\lambda_2(v_k)}\frac{\partial\Lambda(v_k|\bu)}{\partial u_j}=-\frac{1}{g(v_k,\bv_k)}
\bigg((\tvk-\rho_1)\frac{\lambda_1(v_k)}{\lambda_2(v_k)}f(\bv_k,v_k)\nonumber\\
\hphantom{\sum_{j=1}^N \frac{c\gamma_j}{g(v_k,\bu)\lambda_2(v_k)}\frac{\partial\Lambda(v_k|\bu)}{\partial u_j}=}{}-(\vk-\rho_2)f(v_k,\bv_k)+(\rho_1+\rho_2) g(v_k,\bv_k)\lambda_1(v_k)\bigg).\label{sumM}
\end{gather}
If the set $\bv$ satisfies the system of modified Bethe equations \eqref{BE00}, then the r.h.s.\ of~\eqref{sumM} vanishes. Thus, the rows of the matrix~\eqref{Lamder0} are linearly dependent, and hence, the Jacobian of the transfer matrix eigenvalue in~\eqref{Svu} vanishes.
\end{proof}

Setting $\bv=\bu$ in \eqref{Svunew1} we obtain an expression for square of the norm of on-shell modified Bethe vector\footnote{Traditionally $S_\nu^N(\bu,\bu)= \ccc^N(\bu)\bbb^N(\bu)$ is called the square of the norm even if $\ccc^N(\bu)\ne \left(\bbb^N(\bu)\right)^\dagger$.}. Another possibility is to set $\bv=\bu$ in \eqref{Svu}. Then one should resolve singularities in the diagonal elements of~\eqref{Lamder0}. The result is described by the following theorem.

\begin{Theorem}\label{Mtheor3}The square of the norm of on-shell modified Bethe vector has the following determinant representation:
\begin{gather*}%\label{norm1}
S_\nu^N(\bu,\bu)=\frac{(\mu\beta)^N \Delta(\bu)\Delta'(\bu)\lambda^2_2(\bu)}{(\beta\vk+\tvk-2\rho_1)^N}
K_{N,N}^{(-1/\beta)}(\bu|\bth)
\det_N\left(c\frac{\partial\mathcal{Y}(u_k|\bu)}{\partial u_j}\right),
\end{gather*}
where $\mathcal{Y}(z|\bu)$ is given by \eqref{Y-def}.
\end{Theorem}

\begin{proof} It suffices to substitute $v_j=u_j$ in \eqref{Svu} and compare the result with the partial derivatives $c\partial\mathcal{Y}(u_k|\bu)/\partial u_j$. \end{proof}

\subsection{Alternative form of the Bethe equations}\label{SS-IBE}

The initial form of Bethe equations \eqref{BE00} is not always convenient for applications. There exists, however, various alternative forms of these equations. Two of them are described by the following proposition.

\begin{Proposition}\label{P-AlfBE}Inhomogeneous Bethe equations \eqref{BE00} can be written in the forms
\begin{gather}\label{BE7}
\frac{\vk-\rho_2}{f(\bu,\theta_j)}
+(\tvk-\rho_1)\sum_{k=1}^N\frac{f(\theta_k,\bth_k)}{h(\theta_k,\theta_j)}f(\bu,\theta_k)=\vk+\tvk,\qquad j=1,\dots,N,
\end{gather}
or
\begin{gather}\label{BEj}
(\tvk-\rho_1)f(\bu,\theta_j)+(\vk-\rho_2)\sum_{k=1}^N\frac{f(\bth_k,\theta_k)}{h(\theta_j,\theta_k)}\frac1{f(\bu,\theta_k)}=\vk+\tvk,\qquad j=1,\dots,N.
\end{gather}
\end{Proposition}

The proof of this theorem and other forms of Bethe equations are given in Section~\ref{S-AFBE}.

\subsection{On-shell modified Izergin determinant} \label{SS-MIKOS}

Determinant representations for the scalar product \eqref{Svu} contains the modified Izergin determinant $K_{N,N}^{(z)}(\bu|\bth)$, where the set $\bu$ solves Bethe equations \eqref{BE00}. We call it {\it on-shell modified Izergin determinant}. This determinant enjoys an identity described by the following proposition.

\begin{Proposition}\label{P-IdentID}Let a set $\bu$ consist of the roots of Bethe equations \eqref{BE00}. Then
\begin{gather}\label{K-resul0}
K_{N,N}^{(z)}(\bu|\bth)=\prod_{i=1}^N(d_i-z),
\end{gather}
for arbitrary complex $z$. Here $d_i=d^\pm$ for $i=1,\dots,N$, and
\begin{gather}\label{dpm}
d^\pm=\frac{\vk+\tvk\pm\sqrt{(\vk+\tvk)^2-4(\vk-\rho_2)(\tvk-\rho_1)}}{2(\tvk-\rho_1)}.
\end{gather}
\end{Proposition}

The proof of this proposition is given in Section~\ref{S-MIDOS}.

\begin{Remark}Similar property of the on-shell modified Izergin determinant holds for the XXX spin chain with diagonal boundary condition $($see \eqref{K-resul00}$)$.
\end{Remark}

\section{Proof of the alternative forms of Bethe equations\label{S-AFBE}}

In this section we prove Proposition~\ref{P-AlfBE}.

\begin{proof} Let us divide Bethe equations \eqref{BE00} by the product $\lambda_1(u_j)\lambda_2(u_j)$:
\begin{gather*}%\label{BE1}
(\vk-\rho_2)\frac{f(u_j,\bu_j)}{\lambda_1(u_j)}-(\tvk-\rho_1)\frac{f(\bu_j,u_j)}{\lambda_2(u_j)}+(\rho_1+\rho_2) g(u_j,\bu_j)=0, \qquad j=1,\dots,N.
\end{gather*}
Multiplying each of equations by a monic polynomial $P_n(u_j)$ in $u_j$ of degree $n<N$ and taking the sum over $j$ we obtain
\begin{gather*}%\label{BE2}
\sum_{j=1}^N\left(\frac{(\vk-\rho_2)f(u_j,\bu_j)}{\lambda_1(u_j)}-\frac{(\tvk-\rho_1)f(\bu_j,u_j)}{\lambda_2(u_j)}+(\rho_1+\rho_2) g(u_j,\bu_j)\right)P_n(u_j)=0.
\end{gather*}
The sums over $j$ can be written is a contour integral
\begin{gather*}%\label{BE3}
\frac1{2\pi{\rm i}c}\oint_{\Gamma(\bu)}\left(\frac{(\vk-\rho_2)f(z,\bu)}{\lambda_1(z)}+
\frac{(\tvk-\rho_1)f(\bu,z)}{\lambda_2(z)}+(\rho_1+\rho_2) g(z,\bu)\right)P_n(z) \,{\rm d}z=0.
\end{gather*}
Here anticlockwise oriented contour $\Gamma(\bu)$ surrounds the roots of Bethe equations $\bu$ and does not contain any other singularities of the integrand. Substituting here explicit expressions for $\lambda_1(z)$ and $\lambda_2(z)$ \eqref{lambdas1} we obtain
\begin{gather*}%\label{BE4}
\frac{1}{2\pi{\rm i}}\oint_{\Gamma(\bu)}\left(\frac{(\vk-\rho_2)f(z,\bu)}{h(z,\bth)}+(\tvk-\rho_1)f(\bu,z)g(z,\bth)+ (\rho_1+\rho_2)g(z,\bu)\right)P_n(z)\,{\rm d}z=0.
\end{gather*}
The integral can now be taken by the residues outside the integration contour, that is in the points $z=\theta_k$ and $z=\theta_k-c$, $k=1,\dots,N$. Taking into account that $n<N$ we find{\samepage
\begin{gather}\label{BE5}
\sum_{k=1}^N\left((\vk-\rho_2)\frac{P_n(\theta_k-c)}{f(\bu,\theta_k)}+(\tvk-\rho_1)f(\bu,\theta_k)P_n(\theta_k)\right)g(\theta_k,\bth_k)=c^{N-1}(\vk+\tvk)\delta_{n,N-1},
\end{gather}
where we used \eqref{gfh-prop}.}

This is the most general form of Bethe equations. It involves an arbitrary polynomial $P_n(\theta_k)$ of degree $n<N$. Now we can consider several particular cases.

Let us take a system of polynomials
\begin{gather}\label{Pjh}
P^{(j)}_{N-1}(z)=c^{N-1}h(z,\bth_j)=\prod_{\substack{k=1\\ k\ne j}}^N (z-\theta_k+c), \qquad j=1,\dots,N.
\end{gather}
Then equation \eqref{BE5} yields
\begin{gather*}
\sum_{k=1}^N\left(\!(\vk-\rho_2)\frac{h(\theta_k-c,\bth_j)}{f(\bu,\theta_k)}+(\tvk-\rho_1)f(\bu,\theta_k)h(\theta_k,\bth_j)\!\right)g(\theta_k,\bth_k)
 =\vk+\tvk,\!\qquad j=1,\dots,N.%\label{BE6}
\end{gather*}
Using $h(\theta_k-c,\bth_j)=\delta_{jk}/g(\theta_k,\bth_k)$ we arrive at \eqref{BE7}.

Similarly, setting in \eqref{BE5}
\begin{gather}\label{Pjg}
P^{(j)}_{N-1}(z)=\frac{c^{N-1}}{g(z,\bth_j)}=\prod_{\substack{k=1\\ k\ne j}}^N (z-\theta_k), \qquad j=1,\dots,N,
\end{gather}
we arrive at \eqref{BEj}. \end{proof}

One can also consider a choice of $P^{(j)}_{N-1}(z)$ that interpolates between \eqref{Pjh} and \eqref{Pjg}. Let $\bth_A$ and $\bth_B$ be two fixed disjoint subsets of the inhomogeneities such that $\{\bth_A,\bth_B\}=\bth$. Assume that $\bth_B\ne\varnothing$. Let
\begin{gather}\label{Pjgh}
P^{(j)}_{N-1}(z)=c^{N-1}g(z,\theta_j)\frac{h(z,\bth_{A})}{g(z,\bth_{B})}, \qquad \theta_j\in\bth_{B}.
\end{gather}
Then equation \eqref{BE5} takes the form
\begin{gather*}
(\vk-\rho_2)\sum_{\substack{\{\bth_{B_1},\bth_{B_2}\}\vdash\bth_{B}\\ \#\bth_{B_1}=1}}
\frac{f(\bth_{B_2},\bth_{B_1})}{h(\theta_j,\bth_{B_1})}\frac1{f(\bu,\bth_{B_1})}
=\vk+\tvk-(\tvk-\rho_1)f(\bu,\theta_j)f(\theta_j,\bth_{A})\nonumber\\
\qquad{} +(\tvk-\rho_1)\sum_{\substack{\{\bth_{A_1},\bth_{A_2}\}\vdash\bth_{A}\\ \#\bth_{A_1}=1}}
g(\theta_j,\bth_{A_1})f(\bu,\bth_{A_1})f(\bth_{A_1},\bth_{A_2}).%\label{BEjgen}
\end{gather*}
Here in the l.h.s.\ the sum is taken over partitions $\{\bth_{B_1},\bth_{B_2}\}\vdash\bth_{B}$ such that $\#\bth_{B_1}=1$. In the r.h.s.\ the sum is taken over partitions $\{\bth_{A_1},\bth_{A_2}\}\vdash\bth_{A}$ such that $\#\bth_{A_1}=1$.

\begin{Remark}The examples given above do not exhaust all possible forms of the Bethe equations. In particular, on the basis of the polynomial~\eqref{Pjgh} we can construct \begin{gather}\label{Pmm}
P^{(j)}_{N-1}(z)=c^{N-1}\sum_{\{\bth_A,\bth_B\}\vdash\bth}g(z,\theta_j)\frac{h(z,\bth_{A})}{g(z,\bth_{B})}G(\bth_A,\bth_B), \qquad \theta_j\in\bth_{B},
\end{gather}
where $G(\bth_A,\bth_B)$ is some function that depends on the subsets $\bth_A$ and $\bth_B$. An example of the polynomial~\eqref{Pmm} is considered in Appendix~\ref{A-PL1}.
\end{Remark}

\section{Proof of the on-shell modified Izergin determinant properties}\label{S-MIDOS}

In this section we prove Proposition~\ref{P-IdentID}.

\begin{proof}
Let us introduce two $N\times N$ matrices $\Omega(\bth)$ and $Z$:
\begin{gather*}%\label{AZ}
\Omega_{jk}(\bth)=\frac{f(\theta_k,\bth_k)}{h(\theta_k,\theta_j)},\qquad Z_{jk}=\Omega_{jk}(\bth)f(\bu,\theta_k),\qquad j,k=1,\dots,N.
\end{gather*}
Then the system of equations \eqref{BE7} takes the form
\begin{gather*}%\label{BEAZ}
\frac{\vk-\rho_2}{f(\bu,\theta_j)}+(\tvk-\rho_1)\sum_{k=1}^NZ_{jk}=\vk+\tvk,\qquad j=1,\dots,N.
\end{gather*}
Multiplying each equation by $Z_{ij}$ we obtain
\begin{gather*}%\label{BEAZ0}
(\vk-\rho_2)\Omega_{ij}(\bth)+(\tvk-\rho_1)\sum_{k=1}^NZ_{ij}Z_{jk}=(\vk+\tvk)Z_{ij},\qquad i,j=1,\dots,N.
\end{gather*}
Taking the sum over $j$ and using\footnote{Equation \eqref{Asum} can be derived by the contour integral method described in Appendix~\ref{A-SRF}.}
\begin{gather}\label{Asum}
\sum_{j=1}^N\Omega_{ij}(\bth)=1,
\end{gather}
we arrive at
\begin{gather*}%\label{BEAZ1}
\sum_{k=1}^N\bigl((\vk-\rho_2)\delta_{ik} +(\tvk-\rho_1)\big(Z^2\big)_{ik}-(\vk+\tvk)Z_{ik}\bigr)=0,\qquad i=1,\dots,N.
\end{gather*}
Setting here $Z=S^{-1}\mathcal{D}S$, where $\mathcal{D}=\diag(d_1,\dots,d_N)$, we find
\begin{gather*}%\label{BEAZ2}
\sum_{k=1}^N\bigl((\vk-\rho_2)\delta_{ik}+(\tvk-\rho_1)\big(S^{-1}\mathcal{D}^2S\big)_{ik}-(\vk+\tvk)\big(S^{-1}\mathcal{D}S\big)_{ik}\bigr)=0,\qquad i=1,\dots,N.
\end{gather*}
Multiplying this equation from the left by $S$ we arrive at
\begin{gather}\label{BEAZ3}
\sum_{k=1}^N\bigl((\vk-\rho_2) +(\tvk-\rho_1)d_i^2-(\vk+\tvk)d_i\bigr)S_{ik}=0,\qquad i=1,\dots,N.
\end{gather}
Since $S$ is invertible, we conclude that linear combination \eqref{BEAZ3} vanishes if and only if
\begin{gather}\label{BEAZ4}
(\tvk-\rho_1)d_i^2 -(\vk+\tvk)d_i+\vk-\rho_2=0,\qquad i=1,\dots,N.
\end{gather}

Thus, the matrix $Z$ has only two eigenvalues $d^+$ and $d^-$~\eqref{dpm}. It also follows from the above consideration that the on-shell modified Izergin determinant $K_{N,N}^{(z)}(\bu|\bth)$ has the following presentation:
\begin{gather}\label{K-resul}
K_{N,N}^{(z)}(\bu|\bth)=\det_N(Z_{jk}-z\delta_{jk})=\prod_{i=1}^N(d_i-z),
\end{gather}
for arbitrary complex $z$, provided the parameters $\bu$ are on-shell. Thus, \eqref{K-resul0} is proved.
\end{proof}

In particular, setting $z=0$ in \eqref{K-resul} we obtain
\begin{gather}\label{prod-f}
f(\bu,\bth)=\prod_{i=1}^Nd_i.
\end{gather}

Using equation \eqref{BEAZ4} one can find identities for the on-shell modified Izergin determinants with different $z$. For example, \eqref{BEAZ4} yields
\begin{gather*}%\label{BEAZ5}
\frac{\beta(\vk-\rho_2)}{\beta\vk+\tvk-2\rho_1}(d_i-1)\left(d_i+\frac1\beta\right)=d_i\left(d_i-\mu-\frac{\mu-1}\beta\right),\qquad i=1,\dots,N.
\end{gather*}
Taking the product over $i$ we obtain
\begin{gather*}%\label{BEAZ8}
\left(\frac{\beta(\vk-\rho_2)}{\beta\vk+\tvk-2\rho_1}\right)^N\prod_{i=1}^N(d_i-1)\left(d_i+\frac1\beta\right)=\prod_{i=1}^N d_i\left(d_i-\mu-\frac{\mu-1}\beta\right).
\end{gather*}
Then, due to equations \eqref{K-resul}, \eqref{prod-f} we find
\begin{gather}\label{finId}
\left(\frac{\beta(\vk-\rho_2)}{\beta\vk+\tvk-2\rho_1}\right)^N K_{N,N}^{(1)}(\bu|\bth)K_{N,N}^{(-1/\beta)}(\bu|\bth)=f(\bu,\bth)K_{N,N}^{(\mu+(\mu-1)/\beta)}(\bu|\bth).
\end{gather}

Concluding this section we would like to mention that exactly the same method can be used to transform Bethe equations of the XXX spin chain with diagonal boundary condition
\begin{gather*}%\label{UBE}
\vk\frac{f(u_j,\bu_j)}{\lambda_1(u_j)}-\tvk\frac{f(\bu_j,u_j)}{\lambda_2(u_j)}=0, \qquad j=1,\dots,M,
\end{gather*}
where $\# \bu=M\leq N$. Similarly to \eqref{K-resul0} we then obtain
\begin{gather}\label{K-resul00}
K_{M,N}^{(z)}(\bu|\bth)=\prod_{i=1}^N(d_i-z),
\end{gather}
where $d_i\in\{1,\vk/\tvk\}$. Using definition of the modified Izergin determinant \eqref{defKdef2} we can write down \eqref{K-resul00} as
\begin{gather*}%\label{NEWBE}
(1-z)^{N-M}\det_M\left(\delta_{jk}f(u_j,\bth)-z\frac{f(u_j,\bu_j)}{h(u_j,u_k)}\right)=\prod_{i=1}^N(d_i-z).
\end{gather*}

\section{Calculation of the scalar product}\label{S-DPSP}

In this section we prove Theorem~\ref{Mtheor1}. First, we show the equivalence of representations~\eqref{Svunew1} and~\eqref{Svu}. Then we prove~\eqref{Svunew1} starting with equation \eqref{SP-finuv} and specifying it for the case $\#\bu=\#\bv=N$
\begin{gather}
S_\nu^{N}(\bv,\bu)=\mu^{2N}
\sum_{\substack{\{\bv_{\so},\bv_{\st}\}\vdash\bv\\\{\bu_{\so},\bu_{\st}\}\vdash\bu }}
\beta_1^{p_{\st}-q_{\st}}\beta_2^{p_{\so}-q_{\so}}\lambda_2(\bv_{\so})\lambda_1(\bv_{\st}) \lambda_2(\bu_{\st})\lambda_1(\bu_{\so})
f(\bv_{\so}, \bv_{\st}) f(\bu_{\st},\bu_{\so})\nonumber\\
\hphantom{S_\nu^{N}(\bv,\bu)=}{}\times K_{q_{\st},p_{\st}}^{(1/\mu)}(\bu_{\st}|\bv_{\st}) \overline{K}_{q_{\so},p_{\so}}^{(1/\mu)}(\bu_{\so}|\bv_{\so}).\label{SP-finuvB}
\end{gather}
Recall that here the sum is taken over all possible partitions $\{\bv_{\so},\bv_{\st}\}\vdash\bv$ and $\{\bu_{\so},\bu_{\st}\}\vdash\bu$ without any restriction on the cardinalities of the subsets.

We have pointed out in Section~\ref{SS} that the function $S_{N}(\bv|\bu)$ is a polynomial in $\bv$ of degree~$N$ in each~$v_j$. Since this polynomial is symmetric over $\bv$, it has $\binom{2N}{N}$ independent coefficients. Hence, in order to find this polynomial, it is enough to compute it in~$\binom{2N}{N}$ points.

Let us consider an arbitrary partition of the inhomogeneities $\bth=\{\bth_{A},\bth_{B}\}$ with $\#\bth_{A}=n_{A}$ and $\#\bth_{B}=n_{B}$. Suppose that $\bv=\{\bth_{A},\bth_{B}-c\}$. Clearly, there exists exactly $\binom{2N}{N}$ choices of this type. Hence, if we compute the scalar product for all possible $\bv=\{\bth_{A},\bth_{B}-c\}$, then the polynomial is completely determined for an arbitrary set $\bv$. This will be done in Sections~\ref{SS-TSPGU} and~\ref{SS-FTSPOS}.

\subsection{Transformation of Jacobian}\label{S-TJ}

In this subsection we reduce the Jacobian of the transfer matrix eigenvalue to the modified Izergin determinant. We thus prove Theorem~\ref{Mtheor2}.

\begin{proof} Assume that representation \eqref{Svu} holds. Let
\begin{gather*}%\label{Mdef}
M(u_j,v_k)=\frac{c}{g(v_k,\bu)\lambda_2(v_k)}\frac{\partial\Lambda(v_k|\bu)}{\partial u_j},
\end{gather*}
where $\Lambda(v_k|\bu)$ is given by~\eqref{Lam}. Then
\begin{gather}
M(u_j,v_k)=(-1)^{N-1}(\tvk-\rho_1)f(v_k,\bth)h(\bu,v_k)\frac{g(u_j,v_k)}{h(u_j,v_k)} \nonumber\\
\hphantom{M(u_j,v_k)=}{} +(\vk-\rho_2)h(v_k,\bu)\frac{g(v_k,u_j)}{h(v_k,u_j)} +(\rho_1+\rho_2) h(v_k,\bth)g(v_k,u_j).\label{Lamder}
\end{gather}
Recall that here the set $\bv$ consists of arbitrary complex numbers, while the set $\bu$ solves Bethe equations. Therefore, in particular, the matrix elements \eqref{Lamder} have no poles at $v_k=u_j$:
\begin{gather*}\label{Nopoles}
\Res M(u_j,v_k)\Bigr|_{v_k=u_j}=0.
\end{gather*}
Indeed, it is easy to see that this residue is proportional to the Bethe equations \eqref{BE00}.

Obviously, for any non-degenerated $N\times N$ matrix $A$, one has
\begin{gather*}%\label{detM}
\det_N M=\frac{\det_N H}{\det_N A}, \qquad\text{where}\qquad H=MA.
\end{gather*}
Let
\begin{gather*}%\label{A}
A(v_j,\theta_k)=\frac{g(v_j,\bv_j)}{h(v_j,\theta_k)}.
\end{gather*}
Then the determinant of this matrix is proportional to the Cauchy determinant
\begin{gather*}%\label{detA}
\det_N A(v_j,\theta_k)=\frac{\Delta(\bv)}{\Delta(\bth)h(\bv,\bth)},
\end{gather*}
and we obtain
\begin{gather*}%\label{detM1}
\det_N M=\frac{\Delta(\bth)h(\bv,\bth)}{\Delta(\bv)}\det_N H.
\end{gather*}

Let us compute explicitly the entries $H_{jl}=\sum\limits_{k=1}^NM(u_j,v_k)A(v_k,\theta_l)$. Consider an auxiliary contour integral
\begin{gather}\label{CI}
I=\frac1{2\pi{\rm i} c}\oint_{|z|=R\to\infty}\frac{g(z,\bv)}{h(z,\theta_l)}M(u_j,z)\,{\rm d}z.
\end{gather}
It is easy to see that the integrand behaves as $z^{-2}$ as $z\to\infty$. Thus, $I=0$. On the other hand, taking
the sum of the residues within the integration contour we obtain
\begin{gather*}%\label{CIres}
I=H_{jl}+(-1)^N\frac{M(u_j,\theta_l-c)}{h(\bv,\theta_l)}+\frac1c\sum_{k=1}^N\frac{g(\theta_k,\bv)}{h(\theta_k,\theta_l)}\Res M(u_j,z)\Bigr|_{z=\theta_k}.
\end{gather*}
Substituting here explicit representation \eqref{Lamder} for $M(u_j,v_k)$ we find
\begin{gather*}
(-1)^{N-1}H_{jl}=\frac{(\vk-\rho_2)g(u_j,\theta_l)}{g(\theta_l,\bu)h(\bv,\theta_l)h(u_j,\theta_l)}
-(\tvk-\rho_1)\sum_{k=1}^N\frac{g(\theta_k,\bv)}{h(\theta_k,\theta_l)}f(\theta_k,\bth_k)h(\bu,\theta_k)\frac{g(u_j,\theta_k)}{h(u_j,\theta_k)}\nonumber\\
\hphantom{(-1)^{N-1}H_{jl}}{} =\sum_{k=1}^N\frac{(\vk-\rho_2)g(u_j,\theta_k)}{g(\theta_k,\bu)h(\bv,\theta_k)h(u_j,\theta_k)}\!\left(\delta_{kl}
-\frac{\tvk-\rho_1}{\vk-\rho_2}\frac{f(\theta_k,\bth_k)}{h(\theta_k,\theta_l)}f(\bu,\theta_k)f(\bv,\theta_k)\right).%\label{Hjl}
\end{gather*}
Thus
\begin{gather*}%\label{detHXY}
\det_NH=\det_N H^{(1)} \det_N H^{(2)},
\end{gather*}
where
\begin{gather*}%\label{X}
H^{(1)}_{jk}=\frac{(\vk-\rho_2)}{g(\theta_k,\bu)h(\bv,\theta_k)}\frac{g(u_j,\theta_k)}{h(u_j,\theta_k)},
\end{gather*}
and
\begin{gather*}%\label{Y}
H^{(2)}_{jk}=\delta_{jk}-\frac{\tvk-\rho_1}{\vk-\rho_2}\frac{f(\theta_j,\bth_j)}{h(\theta_j,\theta_k)}f(\bu,\theta_j)f(\bv,\theta_j).
\end{gather*}
It is easy to see that $\det\limits_N H^{(1)}$ is proportional to the ordinary Izergin determinant \eqref{OrdIK}
\begin{gather*}%\label{detX}
\det_NH^{(1)} =\frac{(\vk-\rho_2)^N}{g(\bth,\bu)h(\bv,\bth)}\det_N \left(\frac{g(u_j,\theta_k)}{h(u_j,\theta_k)}\right)
=\frac{(-1)^N(\vk-\rho_2)^N}{h(\bv,\bth)f(\bu,\bth)\Delta'(\bu)\Delta(\bth)} K_{N,N}^{(1)}(\bu|\bth),
\end{gather*}
while $\det\limits_N H^{(2)} $ is proportional to the modified Izergin determinant
\begin{gather*}%\label{detY}
\det_NH^{(2)} =(-1)^N\left(\frac{\tvk-\rho_1}{\vk-\rho_2}\right)^N K_{2N,N}^{(\alpha)}(\{\bu,\bv\}|\bth), \qquad \text{where}\qquad
\alpha=\frac{\vk-\rho_2}{\tvk-\rho_1}.
\end{gather*}
Thus, we finally obtain
\begin{gather}\label{detMres}
\det_N M=\frac{(\tvk-\rho_1)^N}{f(\bu,\bth)\Delta'(\bu)\Delta(\bv)} K_{N,N}^{(1)}(\bu|\bth)
K_{2N,N}^{(\alpha)}(\{\bu,\bv\}|\bth).
\end{gather}

We would like to stress that we essentially used the fact that $\bu$ satisfies Bethe equations~\eqref{BE00}. Therefore, $M(u_j,v_k)$ has no pole at $v_k=u_j$. Otherwise the contour integral \eqref{CI} would have an additional contribution. Thus, equation~\eqref{detMres} holds only on-shell, that is $\bu$ satisfies Bethe equations \eqref{BE00}.

Substituting \eqref{detMres} into \eqref{Svu} we obtain
\begin{gather*}%\label{Svunew}
S_\nu^N(\bv,\bu)=\left(\frac{\mu\beta(\tvk-\rho_1)}{\beta\vk+\tvk-2\rho_1}\right)^N \frac{\lambda_2(\bu)\lambda_2(\bv)}{f(\bu,\bth)}
K_{N,N}^{(-1/\beta)}(\bu|\bth) K_{N,N}^{(1)}(\bu|\bth)K_{2N,N}^{(\alpha)}(\{\bu,\bv\}|\bth).
\end{gather*}
Using identity \eqref{finId} we immediately arrive at \eqref{Svunew1}. Thus, representations~\eqref{Svu} and \eqref{Svunew1} are equivalent. \end{proof}

\subsection[Transformation of the scalar product for generic $\bu$]{Transformation of the scalar product for generic $\boldsymbol{\bu}$} \label{SS-TSPGU}

We now turn back to equation \eqref{SP-finuvB}. Using equation \eqref{c-c} and $\lambda_1(u_j)/\lambda_2(u_j)=f(u_j,\bth)$ we obtain
\begin{gather*}
S_\nu^{N}(\bv,\bu)=\mu^{2N}\lambda_2(\bu) \sum_{\substack{\{\bv_{\so},\bv_{\st}\}\vdash\bv\\\{\bu_{\so},\bu_{\st}\}\vdash\bu }}\beta_1^{p_{\st}-q_{\st}}\beta_2^{p_{\so}-q_{\so}}
\left(1-\frac1\mu\right)^{p_{\so}-q_{\so}}\lambda_2(\bv_{\so})\lambda_1(\bv_{\st}) f(\bu_{\so},\bth)\nonumber\\
\hphantom{S_\nu^{N}(\bv,\bu)=}{} \times f(\bv_{\so}, \bv_{\st}) f(\bu_{\st},\bu_{\so})
K_{q_{\st},p_{\st}}^{(1/\mu)}(\bu_{\st}|\bv_{\st}) K_{p_{\so},q_{\so}}^{(1/\mu)}(\bv_{\so}|\bu_{\so}).%\label{SP-finuv-redNN}
\end{gather*}
Now we set $\bv=\{\bth_{A},\bth_{B}-c\}$ and denote the corresponding scalar product by $S_{AB}$. Since $\lambda_2(\theta_j)=0$ and $\lambda_1(\theta_j-c)=0$ for all $j=1,\dots,N$, we see that the sum over partitions $\{\bv_{\so},\bv_{\st}\}\vdash\bv$ is frozen. The only non-vanishing contribution occurs for $\bv_{\so}=\bth_{B}-c$ and $\bv_{\st}=\bth_{A}$. Then we have
\begin{gather}
S_{AB}=\mu^{2N}\lambda_2(\bu) G_{AB} \sum_{\{\bu_{\so},\bu_{\st}\}\vdash\bu } \beta^{n_A-q_{\st}}
\left(1-\frac1\mu\right)^{n_{B}-q_{\so}} f(\bu_{\so},\bth)f(\bu_{\st},\bu_{\so})\nonumber\\
\hphantom{S_{AB}=}{}\times K_{q_{\st},n_{A}}^{(1/\mu)}(\bu_{\st}|\bth_{A}) K_{n_{B},q_{\so}}^{(1/\mu)}(\bth_{B}-c|\bu_{\so}),\label{SPab1}
\end{gather}
where we used $n_A+n_B=q_{\so}+q_{\st}=N$, and
\begin{gather}\label{Gab}
G_{AB}=\frac{\lambda_2(\bth_{B}-c)\lambda_1(\bth_{A})}{f(\bth_{A},\bth_{B})}= (-1)^{Nn_{B}}\frac{h(\bth,\bth_{B})h(\bth_{A},\bth)}{f(\bth_{A},\bth_{B})}.
\end{gather}
Recall also that $\beta=\beta_1/\beta_2$. Our goal now is to calculate the sum over partitions of the set $\bu$.

Using \eqref{Kinv1} we rewrite \eqref{SPab1} as follows:
\begin{gather}
S_{AB}=\mu^{2N}\lambda_2(\bu) G_{AB}\sum_{\{\bu_{\so},\bu_{\st}\}\vdash\bu } \beta^{n_A-q_{\st}}
\left(-\frac1\mu\right)^{n_{B}} f(\bu_{\so},\bth_{A})f(\bu_{\st},\bu_{\so})\nonumber\\
\hphantom{S_{AB}=}{}\times K_{q_{\st},n_{A}}^{(1/\mu)}(\bu_{\st}|\bth_{A}) K_{q_{\so},n_{B}}^{(\mu)}(\bu_{\so}|\bth_{B}).\label{SPab2}
\end{gather}
The next step is to use the second equation~\eqref{K-sumpart} for both modified Izergin determinants in~\eqref{SPab2}:
\begin{gather*}
K_{q_{\st},n_{A}}^{(1/\mu)}(\bu_{\st}|\bth_{A})= \left(1-\frac1\mu\right)^{n_{A}-q_2-q_4}
\sum_{\{\bu_{2},\bu_{4}\}\vdash\bu_{\st}}\left(-\frac1\mu\right)^{q_2} f(\bu_{4},\bth_{A})f(\bu_{2},\bu_{4}),\nonumber\\
K_{q_{\so},n_{B}}^{(\mu)}(\bu_{\so}|\bth_{B})=\left(1-\mu\right)^{n_{B}-q_1-q_3}
\sum_{\{\bu_{1},\bu_{3}\}\vdash\bu_{\so}}(-\mu)^{q_1} f(\bu_{3},\bth_{B})f(\bu_{1},\bu_{3}).%\label{K-sums}
\end{gather*}
Here we use Arabic numeration of the subsets, in order to avoid too cumbersome Roman numerals. We also set $q_j=\#\bu_j$. Substituting these equations into~\eqref{SPab2} we obtain
\begin{gather}
S_{AB}=\mu^{2N}\lambda_2(\bu)G_{AB}\sum_{\{\bu_{1},\bu_{2},\bu_{3},\bu_{4}\}\vdash\bu }\beta^{n_A-q_2-q_4}
\left(-\frac1\mu\right)^{N-q_1-q_4} f(\bu_{3},\bth_{B}) f(\bu_{1},\bth_{A}) \nonumber\\
\hphantom{S_{AB}=}{} \times f(\bu_{3},\bth_{A})f(\bu_{4},\bth_{A})
f(\bu_{2},\bu_{1}) f(\bu_{2},\bu_{3})f(\bu_{4},\bu_{1}) f(\bu_{4},\bu_{3})f(\bu_{2},\bu_{4}) f(\bu_{1},\bu_{3}).\label{SPab3}
\end{gather}
Here the sum is taken over partitions of the set $\bu$ into four subsets $\{\bu_1,\bu_2,\bu_3,\bu_4\}\vdash\bu$.

Let $\bu_{0}=\{\bu_{1},\bu_{4}\}$. Then \eqref{SPab3} transforms as follows:
\begin{gather*}
S_{AB}=\mu^{2N}\lambda_2(\bu)G_{AB}\sum_{\{\bu_{0},\bu_{2},\bu_{3}\}\vdash\bu }\beta^{n_A-q_2}
\left(-\frac1\mu\right)^{N-q_0} f(\bu_{0},\bth_{A}) \nonumber\\
\hphantom{S_{AB}=}{} \times f(\bu_{3},\bth) f(\bu_{2},\bu_{0}) f(\bu_{2},\bu_{3}) f(\bu_{0},\bu_{3})
\sum_{\{\bu_{1},\bu_{4}\}\vdash\bu_{0} } \beta^{-q_4}f(\bu_{4},\bu_{1}),%\label{SPab4}
\end{gather*}
where $q_0=\#\bu_0$. The sum over partitions is now organized into two steps. First, the set $\bu$ is divided into three subsets $\{\bu_0,\bu_2,\bu_3\}\vdash\bu$, and then the subset $\bu_0$ is divided once more as $\{\bu_1,\bu_4\}\vdash\bu_0$. Obviously, the last sum over partitions of the subset $\bu_{0}$ gives $(1+\beta^{-1})^{q_0}$ (see~\cite{BelS18a}), and we arrive at
\begin{gather}
S_{AB}=(-\mu)^{N}\lambda_2(\bu)G_{AB} \sum_{\{\bu_{0},\bu_{2},\bu_{3}\}\vdash \bu}(-\mu)^{q_0}(1+\beta)^{q_0}\beta^{n_A-q_2-q_0}
f(\bu_{0},\bth_{A}) f(\bu_{3},\bth) \nonumber\\
\hphantom{S_{AB}=}{}\times f(\bu_{2},\bu_{0}) f(\bu_{2},\bu_{3}) f(\bu_{0},\bu_{3}) .\label{SPab5}
\end{gather}

Now we combine the subsets $\bu_{0}$ and $\bu_{2}$ into one subset. Let $\bu_{1}=\{\bu_{0},\bu_{2}\}$ with $\#\bu_{1}=q_1$. Then equation~\eqref{SPab5} takes the form
\begin{gather*}
S_{AB}=(-\mu)^{N}\lambda_2(\bu)G_{AB}
\sum_{\{\bu_{1},\bu_{3}\}\vdash\bu } \beta^{n_A-q_1} f(\bu_{3},\bth) f(\bu_{1},\bu_{3})\nonumber\\
\hphantom{S_{AB}=}{}
\times \sum_{\{\bu_{0},\bu_{2}\}\vdash\bu_{1} } (-\mu)^{q_0}(1+\beta)^{q_0} f(\bu_{0},\bth_{A}) f(\bu_{2},\bu_{0}).%\label{SPab6}
\end{gather*}
The sum over partitions $\{\bu_{0},\bu_{2}\}\vdash\bu_{1}$ gives a new modified Izergin determinant due to the second equation \eqref{K-sumpart}, and we find
\begin{gather*}
S_{AB}=(-\mu)^{N}\lambda_2(\bu)G_{AB} \sum_{\{\bu_{1},\bu_{3}\} \vdash\bu} f(\bu_{3},\bth) f(\bu_{1},\bu_{3})\nonumber\\
\hphantom{S_{AB}=}{}\times \frac{\beta^{n_A-q_1}(-\mu-\mu\beta)^{n_A}}{(1-\mu-\mu\beta)^{n_A-q_1}} K_{q_1,n_A}^{(1/(\mu+\mu\beta))}(\bu_{1}|\bth_{A}) .%\label{SPab7}
\end{gather*}

We now develop the modified Izergin determinant $K_{q_1,n_A}^{(1/(\mu+\mu\beta))}(\bu_{1}|\bth_{A})$ via the first equation~\eqref{K-sumpart}, as a sum over partitions $\{\bth_{A_1},\bth_{A_2}\}\vdash\bth_{A}$:
\begin{gather}
S_{AB}=(-\mu)^{N}\lambda_2(\bu)G_{AB} \frac{(\mu+\mu\beta)^{n_A}\beta^{n_A}}{(\mu+\mu\beta-1)^{n_A}}
\sum_{\{\bth_{A_1},\bth_{A_2}\}\vdash\bth_{A} } \sum_{\{\bu_{1},\bu_{3}\}\vdash\bu } \left(\frac{1-\mu}\beta-\mu\right)^{q_1} f(\bu_{3},\bth) \nonumber\\
\hphantom{S_{AB}=}{}
\times \left(-\frac1{\mu+\mu\beta}\right)^{n_{A_2}} f(\bu_{1},\bu_{3})f(\bu_{1},\bth_{A_1})f(\bth_{A_1},\bth_{A_2}),\label{SPab8}
\end{gather}
where $n_{A_1}=\#\bth_{A_1}$, $n_{A_2}=\#\bth_{A_2}=n_A-n_{A_1}$. Let $\{\bth_{B},\bth_{A_2}\}=\bth_{C}$. Then \eqref{SPab8} can be written as follows:
\begin{gather}
S_{AB}=(-\mu)^{N}\lambda_2(\bu)G_{AB} \frac{(\mu+\mu\beta)^{n_A}\beta^{n_A}}{(\mu+\mu\beta-1)^{n_A}}
\sum_{\{\bth_{A_1},\bth_{A_2}\}\vdash\bth_{A} } \left(-\frac1{\mu+\mu\beta}\right)^{n_{A_2}} f(\bu,\bth_{A_1})f(\bth_{A_1},\bth_{A_2})\nonumber\\
\hphantom{S_{AB}=}{}\times\sum_{\{\bu_{1},\bu_{3}\}\vdash\bu } \left(\frac{1-\mu}\beta-\mu\right)^{q_1} f(\bu_{3},\bth_{C}) f(\bu_{1},\bu_{3}).\label{SPab9}
\end{gather}
The sum over partitions $\{\bu_{1},\bu_{3}\}\vdash\bu $ gives again the modified Izergin determinant due to the second equation \eqref{K-sumpart}:
\begin{gather*}
\sum_{\{\bu_{1},\bu_{3}\}\vdash\bu } \left(\frac{1-\mu}\beta-\mu\right)^{q_1} f(\bu_{3},\bth_{C}) f(\bu_{1},\bu_{3})\nonumber\\
\qquad{} =\left(\frac{1-\mu}\beta-\mu\right)^{N-n_B-n_{A_2}}K_{N,N-n_{A_1}}^{(\mu+(\mu-1)/\beta)}(\bu|\{\bth\setminus\bth_{A_1}\}).%\label{sumlast}
\end{gather*}
Substituting this into \eqref{SPab9} we finally arrive at
\begin{gather}
S_{AB}=\frac{(-\mu)^{N}\lambda_2(\bu)G_{AB}\beta^{n_A}}{(1-\mu-\mu\beta)^{n_A}} \sum_{\{\bth_{A_1},\bth_{A_2}\}\vdash\bth_{A} } \xi^{n_{A_1}}
f(\bu,\bth_{A_1})f(\bth_{A_1},\bth_{A_2})\nonumber\\
\hphantom{S_{AB}=}{} \times K_{N,N-n_{A_1}}^{(\mu+(\mu-1)/\beta)}(\bu|\{\bth\setminus\bth_{A_1}\}),\label{SPab10}
\end{gather}
where
\begin{gather}\label{xi}
\xi=\frac\mu\beta (\beta+1)^2(\mu-1).
\end{gather}

Thus, we computed the sum over partitions of the set $\bu$. However, instead of the original partitions, we obtained a new sum over partitions of the inhomogeneities $\bth$.

\subsection[Further transformation of the scalar product for $\bu$ on-shell]{Further transformation of the scalar product for $\boldsymbol{\bu}$ on-shell}\label{SS-FTSPOS}

Observe that up to now the parameters $\bu$ were arbitrary complex numbers. Now we require them to be on-shell. Then due to Theorem~\ref{MT} we have
\begin{gather*}%\label{Mainthmp}
\xi^{n_{A_1}}f(\bu,\bth_{A_1})K^{(\mu+(\mu-1)/\beta)}_{N,N-n_{A_1}}(\bu|\{\bth\setminus\bth_{A_1}\})
=\eta^{n_{A_1}} K^{(\mu+(\mu-1)/\beta)}_{N,N}(\bu|\bth) K^{(1/\eta)}_{N,n_{A_1}}(\bu|\bth_{A_1}),
\end{gather*}
where $\eta=\big(\frac{\mu-1}\beta+\mu\big)\frac{\tvk-\rho_1}{\vk-\rho_2}$ (see \eqref{eta}). Then \eqref{SPab10} takes the form
\begin{gather}
S_{AB}=\frac{(-\mu)^{N}\lambda_2(\bu)G_{AB}\beta^{n_A}}{(1-\mu-\mu\beta)^{n_A}}K^{(\mu+(\mu-1)/\beta)}_{N,N}(\bu|\bth)\nonumber\\
\hphantom{S_{AB}=}{}\times
\sum_{\{\bth_{A_1},\bth_{A_2}\}\vdash \bth_{A}} \eta^{n_{A_1}}f(\bth_{A_1},\bth_{A_2}) K^{(1/\eta)}_{N,n_{A_1}}(\bu|\bth_{A_1}).\label{Q1}
\end{gather}
The remaining sum over partitions now can be computed via \eqref{sun-Kf}:
\begin{gather*}%\label{applSF}
\sum_{\{\bth_{A_1},\bth_{A_2}\}\vdash\bth_{A} } \eta^{n_{A_1}} f(\bth_{A_1},\bth_{A_2}) K^{(1/\eta)}_{N,n_{A_1}}(\bu|\bth_{A_1})=\eta^{n_{A}}K^{(0)}_{N,n_{A}}(\bu|\bth_{A}) =\eta^{n_{A}}f(\bu,\bth_{A}).
\end{gather*}
Substituting this into \eqref{Q1} we immediately arrive at
\begin{gather}\label{Q4}
S_{AB}=(-1)^{n_B}\mu^{N}\alpha^{-n_{A}}\lambda_2(\bu)G_{AB}f(\bu,\bth_{A})K^{(\mu+(\mu-1)/\beta)}_{N,N}(\bu|\bth),
\end{gather}
where $\alpha$ is given by \eqref{alpha0}.

It remains to compare this result with representation \eqref{Svunew1} at $\bv=\{\bth_{A},\bth_{B}-c\}$. For this we should calculate $\lambda_2(\bv) K_{2N,N}^{(\alpha)}(\{\bu,\bv\}|\bth)$ for given values of $\bv$. Replacing $\bth_{A}$ with $\bth_{A}+\epsilon$ we obtain
\begin{gather*}
\lim_{\epsilon\to 0}\lambda_2(\bth_{A}+\epsilon)\lambda_2(\bth_{B}-c)K_{2N,N}^{(\alpha)}(\{\bu,\bth_{A}+\epsilon,\bth_{B}-c\}|\bth)\nonumber\\
\qquad{} =(-\alpha)^{n_B}\lim_{\epsilon\to 0}\frac{(-1)^{Nn_B}h(\bth,\bth_{B})}
{g(\bth_A,\bth_A+\epsilon) g(\bth_A,\bth_B)}K_{2N,N}^{(\alpha)}(\{\bu,\bth_{A}+\epsilon\}|\bth_{A}),%\label{limK1}
\end{gather*}
where we used representation \eqref{lambdas1} for the function $\lambda_2(u)$ and reduction property \eqref{Kz}. Using now \eqref{mresK} we find
\begin{gather*}
\lim_{\epsilon\to 0}\lambda_2(\bth_{A}+\epsilon)\lambda_2(\bth_{B}-c)K_{2N,N}^{(\alpha)}(\{\bu,\bth_{A}+\epsilon,\bth_{B}-c\}|\bth)\nonumber\\
\qquad{}=(-\alpha)^{n_B}\frac{(-1)^{Nn_B}h(\bth,\bth_{B})h(\bth_A,\bth_A)}{ g(\bth_A,\bth_B)} f(\bu,\bth_{A})=(-\alpha)^{n_B}G_{AB}f(\bu,\bth_{A}),%\label{limK2}
\end{gather*}
where $G_{AB}$ is given by \eqref{Gab}. Substituting this expression for $\lambda_2(\bv)K_{2N,N}^{(\alpha)}(\{\bu,\bv\}|\bth)$ at $\bv=\{\bth_{A},\bth_{B}-c\}$ into~\eqref{Svunew1} we arrive at~\eqref{Q4}.

Thus, for an arbitrary partition $\bth=\{\bth_{A},\bth_{B}\}$, the values of the polynomial $S_\nu^{N}(\bv,\bu)$ at $\bv=\{\bth_{A},\bth_{B}-c\}$ are given by equation~\eqref{Svunew1}. As we discussed above, this means that $S_\nu^{N}(\bv,\bu)$ is given by~\eqref{Svunew1} for arbitrary complex~$\bv$.

\section{Conclusion}

In this paper we proved the determinant representation for the OFS-ONS scalar product conjectured in~\cite{BelP} within the framework of the MABA. Similarly to the known determinant formulas, this representation contains the Jacobian of the transfer matrix eigenvalue. However, due to the peculiarities of the MABA this Jacobian can be reduced to the modified Izergin determinant. This possibility arises due to the fact that the number of the Bethe parameters in the modified on-shell Bethe vector coincides with the number of sites of the chain.

In spite of representation \eqref{Svu} formally looks very similar to the one obtained in \cite{KitMT99, Sla89} for the XXX chain with periodic boundary conditions, the proof is very different. Recall that the determinant formula \cite{Sla89} for the ABA solvable models holds for the most general case when the Hilbert space of the monodromy matrix entries is not specified. In particular, it can be infinite-dimensional. On the contrary, the proof given in this paper is essentially based on the fact that we work with a finite-dimensional representation of the $RTT$ algebra. In its turn, this yields a set of very specific properties of the on-shell Izergin determinant. The latter also is an essential part of our proof.

All this does not mean, however, that there is no other proof of the representation \eqref{Svu}, which would be more general and would not rely on the properties of a particular model. The are planning to publish such a proof in our forthcoming publication.

\appendix

\section{Properties of the modified Izergin determinant}\label{A-PMID}

In this section we give a list of properties of the modified Izergin determinant introduced in Section~\ref{A-MDID}. The proofs can be found in~\cite{BelSV18c}. In all the formulas listed below $\bu$, $\bv$, and $z$ are arbitrary complex numbers such that $\#\bu=n$ and $\#\bv=m$.

We begin with a direct relation between modified Izergin determinant and its conjugated:
\begin{gather}\label{c-c}
\overline{K}_{n,m}^{(z)}(\bu|\bv)=(1-z)^{m-n}K_{m,n}^{(z)}(\bv|\bu), \qquad K^{(z)}_{n,m}(-\bu|-\bv)=\overline{K}^{(z)}_{n,m}(\bu|\bv).
\end{gather}
Due to this relation all other properties are given for $K^{(z)}_{n,m}(\bu|\bv)$ only.
\begin{itemize}\itemsep=0pt
\item For arbitrary complex $w$
\begin{gather*}%\label{Elprop}
K^{(z)}_{n,m}(\bu-w|\bv)=K^{(z)}_{n,m}(\bu|\bv+w).
\end{gather*}
\item The modified Izergin determinant has the following initial conditions:
\begin{gather}\label{K0}
K_{n,0}^{(z)}(\bu|\varnothing)=1, \qquad K_{0,n}^{(z)}(\varnothing|\bv)=(1-z)^n.
\end{gather}
\item If one of the arguments goes to infinity, then
\begin{gather*}
K_{n,m}^{(z)}(\bu|\bv)\Bigr|_{u_k\to\infty}=K_{n-1,m}^{(z)}(\bu_k|\bv),\nonumber\\
K_{n,m}^{(z)}(\bu|\bv)\Bigr|_{v_k\to\infty}=(1-z)K_{n,m-1}^{(z)}(\bu|\bv_k).%\label{Kvinf}
\end{gather*}
\item For arbitrary complex $w$
\begin{gather}\label{Kz}
K_{n+1,m+1}^{(z)}(\{\bu,w-c\}|\{\bv,w\})=-z K_{n,m}^{(z)}(\bu|\bv).
\end{gather}
\item The modified Izergin determinant has the following representations as sums over partitions
\begin{gather}
K_{n,m}^{(z)}(\bu|\bv)=\sum_{\{\bv_{\so},\bv_{\st}\}\vdash\bv}(-z)^{\#\bv_{\st}} f(\bu,\bv_{\so})f(\bv_{\so},\bv_{\st}),\nonumber\\
K_{n,m}^{(z)}(\bu|\bv)=(1-z)^{m-n}\sum_{\{\bu_{\so},\bu_{\st}\}\vdash\bu}(-z)^{\#\bu_{\so}} f(\bu_{\st},\bv)f(\bu_{\so},\bu_{\st}).\label{K-sumpart}
\end{gather}
Here the sum is taken over all partitions $\{\bv_{\so},\bv_{\st}\}\vdash\bv$ in the first equation, while in the second equation, the sum is taken over all partitions $\{\bu_{\so},\bu_{\st}\}\vdash\bu$.
\item The order of the sets $\bu$ and $\bv$ can be changed via
\begin{gather}\label{Kinv1}
K_{n,m}^{(z)}(\bu|\bv+c)=\frac{(-z)^{n}(1-z)^{m-n} }{f(\bv,\bu)}K_{m,n}^{(1/z)}(\bv|\bu).
\end{gather}
\item The modified Izergin determinant has simple poles at $u_j=v_k$. The residue at $u_n=v_m$ is given by
\begin{gather}\label{resK}
K_{n,m}^{(z)}(\bu|\bv)\Bigr|_{u_n\to v_m}=g(u_n,v_m)f(\bu_n,u_n)f(v_m,\bv_m)K_{n-1,m-1}^{(z)}(\bu_n|\bv_m)+{\rm reg},
\end{gather}
where ${\rm reg}$ means regular part. It follows from \eqref{resK} that if $\#\bw=\#\bw'=l$, then
\begin{gather}\label{mresK}
\lim_{\bw'\to\bw}\frac{K_{n+l,m+l}^{(z)}(\bu|\bv)}{f(\bw',\bw)}=f(\bu,\bw)f(\bw,\bv)K_{n,m}^{(z)}(\bu|\bv).
\end{gather}
\item
Let $\bu$ and $\bv$ be sets of arbitrary complex numbers such that $\#\bu=n$ and $\#\bv=m$. Then
\begin{gather}\label{sun-Kf}
\sum_{\{\bv_{\so},\bv_{\st}\}\vdash\bv}z_1^{l_{\st}} K^{(z_2)}_{n,l_{\so}}(\bu|\bv_{\so}) f(\bv_{\so},\bv_{\st})=K_{n,m}^{(z_2-z_1)}(\bu|\bv).
\end{gather}
Here $l_{\st}=\#\bv_{\st}$. The sum is taken with respect to all partitions $\{\bv_{\so},\bv_{\st}\}\vdash\bv$. There is no any restriction on the cardinalities of the subsets.
\end{itemize}

\section{Identities for rational functions}\label{A-IRF}

\subsection{Sums of rational functions}\label{A-SRF}

To prove \eqref{sumM} we use the following summation formulas:
\begin{gather}
\sum_{j=1}^N \frac{g(u_j,v_k)}{h(u_j,v_k)}\gamma_j = \frac{h(\bv,v_k)}{h(\bu,v_k)},\nonumber\\
\sum_{j=1}^N \frac{g(v_k,u_j)}{h(v_k,u_j)}\gamma_j = -\frac{h(v_k,\bv)}{h(v_k,\bu)},\nonumber\\
\sum_{j=1}^N g(v_k,u_j)\gamma_j =-1.\label{sumform}
\end{gather}
All these formulas can be obtained by calculation of special contour integrals. For example, consider an integral
\begin{gather*}
I=\frac1{2\pi{\rm i}c}\oint_{|z|=R\to\infty} \frac{g(z,v_k)}{h(z,v_k)}\frac{g(z,\bu)}{g(z,\bv)} \,{\rm d}z \nonumber\\
\hphantom{I}{} = \frac1{2\pi{\rm i}}\oint_{|z|=R\to\infty} \frac{c\,{\rm d}z}{(z-v_k)(z-v_k+c)}\prod_{i=1}^N\frac{z-v_i}{z-u_i}.%\label{contI}
\end{gather*}
The integration is taken over anticlockwise oriented contour around infinity. Clearly, $I=0$, as the integrand behaves as $z^{-2}$ at $z\to\infty$. On the other hand, the integral is equal to the sum of residues within the contour. The residues at $z=u_j$ give the sum in the l.h.s.\ of the first equation \eqref{sumform}. One more contribution comes from the residue at $z=v_k-c$. Thus, we obtain
\begin{gather*}%\label{Ires}
I=0=\sum_{j=1}^N \frac{g(u_j,v_k)}{h(u_j,v_k)}\gamma_j-\prod_{i=1}^N\frac{v_k-v_i-c}{v_k-u_i-c}.
\end{gather*}
This immediately implies the first identity \eqref{sumform}. Other identities can be proved exactly in the same manner.

Using now representation \eqref{Lamder0} we obtain
\begin{gather*}
\sum_{j=1}^N\frac{c\gamma_j}{g(v_k,\bu)\lambda_2(v_k)}\frac{\partial\Lambda(v_k|\bu)}{\partial u_j}
=(-1)^{N-1}(\tvk-\rho_1)\frac{\lambda_1(v_k)}{\lambda_2(v_k)}h(\bv,v_k) \nonumber\\
\hphantom{\sum_{j=1}^N\frac{c\gamma_j}{g(v_k,\bu)\lambda_2(v_k)}\frac{\partial\Lambda(v_k|\bu)}{\partial u_j}=}{} -(\vk-\rho_2)h(v_k,\bv) -(\rho_1+\rho_2) \lambda_1(v_k),%\label{sumM2}
\end{gather*}
leading to \eqref{sumM}.

\subsection{Matrices with rational elements}\label{A-MRE}

Let $\bar x=\{x_1,\dots,x_n\}$ and let $\Omega(x)$ be an $n\times n$ matrix with the elements
\begin{gather*}%\label{Omx}
\Omega_{jk}(\bar x)=\frac{f(\bar x_k,x_k)}{h(x_j,x_k)}.
\end{gather*}
Let us prove that the inverse matrix has the following entries:
\begin{gather*}%\label{OImx}
\bigl(\Omega^{-1}(\bar x)\bigr)_{jk}=\frac{f(x_k,\bar x_k)}{h(x_k,x_j)}=\Omega_{jk}(\bar x)\Bigr|_{c\to -c}.
\end{gather*}
For this we consider the product of these two matrices. Let
\begin{gather*}%\label{GOmOm}
G_{jk}=\sum_{l=1}^n \frac{f(\bar x_l,x_l)}{h(x_j,x_l)} \frac{f(x_k,\bar x_k)}{h(x_k,x_l)}.
\end{gather*}
Consider an auxiliary contour integral
\begin{gather*}%\label{IG}
I=\frac1{2\pi{\rm i} c}\oint_{|z|=R\to\infty}\frac{f(\bar x,z)}{h(x_j,z)} \frac{f(x_k,\bar x_k)}{h(x_k,z)} \,{\rm d}z,
\end{gather*}
where the anticlockwise oriented integration contour is around infinite point. Since the integrand behaves as $z^{-2}$ at $z\to\infty$, we have $I=0$. On the other hand, this integral is equal to the sum of the residues within the integration contour. The contribution of the poles at $z=x_l$, $l=1,\dots,n$ gives $-G_{jk}$. One more contribution comes from the pole at $z=x_k+c$ at $j=k$. Thus, we obtain
\begin{gather*}%\label{enpr}
I=0=-G_{jk}+\delta_{jk}f(x_k,\bar x_k)f(\bar x_k,x_k+c)=-G_{jk}+\delta_{jk},
\end{gather*}
where we used \eqref{gfh-prop}. Thus, $G_{jk}=\delta_{jk}$, what ends the proof.

\section{Relations between on-shell modified Izergin determinants\label{A-ROSID}}

\begin{Theorem}\label{MT} Consider an arbitrary partition of inhomogeneities $\{\bth_{A},\bth_{B}\}\vdash\bth$ such that $\#\bth_{A}=n_A$ and $\bth_{B}=N-n_A$. Let $\bu$ be on-shell, that is the set $\bu$ satisfies Bethe equations~\eqref{BE00}. Then
\begin{gather*}%\label{Mainthm}
f(\bu,\bth_{A})K^{(\mu+(\mu-1)/\beta)}_{N,N-n_A}(\bu|\bth_{B})=\left(\frac\eta\xi\right)^{n_A}
K^{(\mu+(\mu-1)/\beta)}_{N,N}(\bu|\bth) K^{(1/\eta)}_{N,n_A}(\bu|\bth_{A}),
\end{gather*}
where
\begin{gather}\label{eta}
\eta=\left(\frac{\mu-1}\beta+\mu\right)\frac{\tvk-\rho_1}{\vk-\rho_2}.
\end{gather}
and $\xi$ is given by \eqref{xi}.
\end{Theorem}

The proof of Theorem~\ref{MT} is based on two preparatory lemmas.

\begin{Lemma}\label{summax} Let $\bth_A$ and $\bth_B$ be two fixed partitions of the set $\bth$ such that $\#\bth_B=n_B>0$. Let $\theta_j\in\bth_{B}$. Then for arbitrary complex $\gamma$ the following identity holds:
\begin{gather*}
\sum_{\substack{\{\bth_{B_1},\bth_{B_2}\}\vdash\bth_{B}\\ \#\bth_{B_1}=1}}
\frac{f(\bth_{B_2},\bth_{B_1})}{h(\theta_j,\bth_{B_1})}K_{N,n_B-1}^{(\gamma)}(\bu|\bth_{B_2})\nonumber\\
\qquad{} =\frac1{\gamma}\Bigl(f(\bu,\theta_j)K_{N,n_B-1}^{(\gamma)}(\bu|\{\bth_{B}\setminus\theta_j\}) -K_{N,n_B}^{(\gamma)}(\bu|\bth_{B})\Bigr).%\label{sumKb}
\end{gather*}
Here $\bu$ and $\bth$ are arbitrary complex numbers. In particular, we do not require the set $\bu$ to be on-shell.
\end{Lemma}

\begin{Lemma}\label{summin}Let $\bth_A$ and $\bth_B$ be two fixed partitions of the set $\bth$ such that $\#\bth_A=n_A<N$. Let $\theta_j\in\bth_{B}$ and $\bu$ be on-shell. Then
\begin{gather}
\eta\sum_{\substack{\{\bth_{B_1},\bth_{B_2}\}\vdash\bth_{B}\\ \#\bth_{B_1}=1}}
\frac{f(\bth_{B_2},\bth_{B_1})}{h(\theta_s,\bth_{B_1})}\frac{K_{N,n_A+1}^{(1/\eta)}(\bu|\{\bth_{A},\bth_{B_1}\})}{f(\bu,\bth_{B_1})}\nonumber\\
\qquad =\frac\beta{\mu\beta+\mu-1}\Bigl(\eta K_{N,n_A+1}^{(1/\eta)}(\bu|\{\bth_{A},\theta_{s}\})-\xi K_{N,n_A}^{(1/\eta)}(\bu|\bth_{A})\Bigr).\label{sumKs}
\end{gather}
\end{Lemma}

The proofs of these lemmas respectively are given in Appendices~\ref{A-PL2} and~\ref{A-PL1}.

\begin{proof}[Proof of Theorem~\ref{MT}] The proof partly uses induction over $n_A$. For $n_A=0$ the statement of the theorem is obvious. Assume that it is valid for some $n_A\ge 0$:
\begin{gather*}%\label{MTn1}
\xi^{n_A}f(\bu,\bth_{A})K^{(\mu+(\mu-1)/\beta)}_{N,n_B}(\bu|\bth_{B})=\eta^{n_A} K^{(\mu+(\mu-1)/\beta)}_{N,N}(\bu|\bth) K^{(1/\eta)}_{N,n_A}(\bu|\bth_{A}),
\end{gather*}
Here $\bth_{A}$ and $\bth_{B}$ are arbitrary subsets of $\bth$ with cardinalities $\#\bth_{A}=n_A$ and $\#\bth_{B}=n_B=N-n_A$ respectively. We assume that $K^{(\mu+(\mu-1)/\beta)}_{N,n_B}(\bu|\bth_{B})\ne 0$ for generic inhomogeneities $\bth$ and generic twist parameters.

Now we make new subsets. We fix some $\theta_\ell\in\bth_B$ and consider subsets $\{\bth_{A},\theta_\ell\}$ and $\{\bth_{B}\setminus\theta_\ell\}$. Let us introduce $Y(\theta_\ell)$ as
\begin{gather}
Y(\theta_\ell) = \xi^{n_A+1}f(\bu,\bth_{A})K^{(\mu+(\mu-1)/\beta)}_{N,n_B-1}(\bu|\{\bth_{B}\setminus\theta_\ell\})\nonumber\\
\hphantom{Y(\theta_\ell) =}{} -\eta^{n_A+1}
K^{(\mu+(\mu-1)/\beta)}_{N,N}(\bu|\bth)\frac{K^{(1/\eta)}_{N,n_A}(\bu|\{\bth_{A},\theta_\ell\})}{f(\bu,\theta_\ell)}.\label{MTn2}
\end{gather}
Obviously, if we prove that $Y(\theta_\ell)=0$, then we prove Theorem~\ref{MT} for $n=n_A+1$.

Let $\Omega$ be an $n_B\times n_B$ matrix with the entries
\begin{gather*}%\label{Omjk}
\Omega_{jk}=\frac{f(\bth^B_{k},\theta^B_k)}{h(\theta^B_j,\theta^B_k)}.
\end{gather*}
Here we temporary used a superscript for the subset $\bth_B=\bth^B$. Respectively, $\theta^B_k$ is $k$-th element of the subset $\bth_B$. We also used $\bth^B_{k}= \bth^B\setminus\theta^B_k$. It can be easily checked (see Appendix~\ref{A-MRE}) that
\begin{gather*}%\label{iOmjk}
\big(\Omega^{-1}\big)_{jk}=\frac{f(\theta^B_k,\bth^B_k)}{h(\theta^B_k,\theta_j)}.
\end{gather*}

Let us compute the action of $\Omega$ on the vector $Y$ \eqref{MTn2}. Obviously
\begin{gather*}%\label{ActOMY}
\sum_{\ell=1}^{n_B}\Omega_{j\ell}Y(\theta_\ell)=
\sum_{\substack{\{\bth_{B_1},\bth_{B_2}\}\vdash\bth_{B}\\ \#\bth_{B_1}=1}}
\frac{f(\bth_{B_2},\bth_{B_1})}{h(\theta_j,\bth_{B_1})}Y(\bth_{B_1}).
\end{gather*}
Thus, we can use the results of Lemmas~\ref{summax} and~\ref{summin}. Simple straightforward calculation gives
\begin{gather*}
\sum_{\ell=1}^{n_B}\Omega_{j\ell}Y(\theta_\ell)=\frac{\beta f(\bu,\theta_j)}{\mu\beta+\mu-1}Y(\theta_j)-\frac{\beta \xi}{\mu\beta+\mu-1}\nonumber\\
\qquad{}\times \Bigl\{\xi^{n_A}f(\bu,\bth_{A})K^{(\mu+(\mu-1)/\beta)}_{N,n_B}(\bu|\bth_{B})-
\eta^{n_A} K^{(\mu+(\mu-1)/\beta)}_{N,N}(\bu|\bth) K^{(1/\eta)}_{N,n_A}(\bu|\bth_{A})\Bigr\}.%\label{actOm0}
\end{gather*}
The term in the second line vanishes due to the induction assumption. Hence,
\begin{gather*}%\label{actOm}
\sum_{\ell=1}^{n_B} \Omega_{j\ell}Y(\theta_\ell)=\frac{\beta f(\bu,\theta_j)}{\mu\beta+\mu-1}Y(\theta_j).
\end{gather*}
or equivalently,
\begin{gather*}%\label{actOm1}
\sum_{\ell=1}^{n_B} \left(\left(\mu+\frac{\mu-1}\beta\right)\Omega_{j\ell}-\delta_{j\ell}f(\bu,\theta_j)\right)Y(\theta_\ell)=0.
\end{gather*}
Multiplying this equation from the left by $\big(\Omega^{-1}\big)_{ij}$ and taking the sum over $\theta_j\in\bth_B$ we arrive at
\begin{gather*}%\label{actOm2}
\sum_{\ell=1}^{n_B}\left(f(\bu,\theta_i)\big(\Omega^{-1}\big)_{i\ell}-\left(\mu+\frac{\mu-1}\beta\right)\delta_{i\ell}\right)Y(\theta_\ell)=0.
\end{gather*}
It is easy to see that
\begin{gather*}%\label{detsys}
\det\left(f(\bu,\theta_i)\big(\Omega^{-1}\big)_{i\ell}-\left(\mu+\frac{\mu-1}\beta\right)\delta_{i\ell}\right) =K_{N,n_B}^{(\mu+(\mu-1)/\beta)}(\bu|\bth_B).
\end{gather*}
Due to the induction assumption $K_{N,n_B}^{(\mu+(\mu-1)/\beta)}(\bu|\bth_B)\ne 0$. Then $Y(\theta_\ell)=0$, and hence, Theorem~\ref{MT} is proved for $n=n_A+1$.
\end{proof}

\section{Proof of preparatory lemmas}\label{A-PPL}

\subsection{Proof of Lemma~\ref{summax}}\label{A-PL2}
Let
\begin{gather*}%\label{F-sumKb1}
F=\sum_{\substack{\{\bth_{B_1},\bth_{B_2}\}\vdash\bth_{B}\\ \#\bth_{B_1}=1}}
\frac{f(\bth_{B_2},\bth_{B_1})}{h(\theta_j,\bth_{B_1})}K_{N,n_B-1}^{(\gamma)}(\bu|\bth_{B_2}).
\end{gather*}
Consider an auxiliary contour integral
\begin{gather*}%\label{F-int}
I=\frac1{2\pi{\rm i}c}\oint_{|z|=R\to\infty}\frac{f(\bth_{B},z)}{h(\theta_j,z)}K_{N+1,n_B}^{(\gamma)}(\{\bu,z-c\}|\bth_{B}).
\end{gather*}
On the one hand, this integral is equal to the residue at infinity:
\begin{gather}\label{Il}
I=-K_{N,N-1}^{(\gamma)}(\bu|\bth_{B}).
\end{gather}
On the other hand, this integral is equal to the sum of residues within the contour. The sum in the points $z=\theta_k\in\bth_{B}$
gives $\gamma F$ (due to \eqref{Kz}). One more pole occurs at $z=\theta_j+c$. Using~\eqref{resK} we find
\begin{gather}\label{Ir}
I=\gamma F- f(\bu,\theta_j)K_{N,n_B-1}^{(\gamma)}(\bu|\{\bth_{B}\setminus\theta_j\}).
\end{gather}
Comparing \eqref{Ir} and \eqref{Il} we arrive at the statement of the lemma.

\subsection{Proof of Lemma~\ref{summin}\label{A-PL1}}

Let $\bth_{B_1}=\theta_{l}$ in \eqref{sumKs}. Consider $K_{N,n_A+1}^{(1/\eta)}(\bu|\{\bth_{A},\theta_{l}\})$ as a function of $\theta_{l}$. Using the first equation~\eqref{K-sumpart} we have
\begin{gather*}%\label{K-pres}
K_{N,n_A+1}^{(1/\eta)}(\bu|\{\bth_{A},\theta_{l}\})=\sum_{\{\bth_{\so},\bth_{\st}\}\vdash\{\bth_{A},\theta_{l}\}}
\left(-\frac1\eta\right)^{\#\bth_{\st}} f(\bu,\bth_{\so})f(\bth_{\so},\bth_{\st}).
\end{gather*}
Clearly, either $\theta_l\in\bth_{\so}$ or $\theta_l\in\bth_{\st}$. In the first case we set $\bth_{\so}=\{\theta_l,\bth_{A_1}\}$ and $\bth_{\st}=\bth_{A_2}$. In the second case we set $\bth_{\so}=\bth_{A_1}$ and $\bth_{\st}=\{\theta_l,\bth_{A_2}\}$. Thus, we obtain
\begin{gather*}
\frac{K_{N,n_A+1}^{(1/\eta)}(\bu|\{\bth_{A},\theta_{l}\})}{f(\bu,\theta_l)}= \sum_{\{\bth_{A_1},\bth_{A_2}\}\vdash\bth_{A}}
\left(-\frac1\eta\right)^{\#\bth_{A_2}} f(\bu,\bth_{A_1})f(\bth_{A_1},\bth_{A_2})\nonumber\\
\hphantom{\frac{K_{N,n_A+1}^{(1/\eta)}(\bu|\{\bth_{A},\theta_{l}\})}{f(\bu,\theta_l)}=}{}\times \left(f(\theta_l,\bth_{A_2})
-\frac1\eta\frac{f(\bth_{A_1},\theta_l)}{f(\bu,\theta_l)}\right).%\label{K-pres1}
\end{gather*}
Substituting this to the l.h.s.\ of~\eqref{sumKs} we have
\begin{gather*}%\label{sumKslem}
\eta\sum_{\substack{\{\bth_{B_1},\bth_{B_2}\}\vdash\bth_{B}\\ \#\bth_{B_1}=1}}
\frac{f(\bth_{B_2},\bth_{B_1})}{h(\theta_s,\bth_{B_1})}
\frac{K_{N,n_A+1}^{(1/\eta)}(\bu|\{\bth_{A},\bth_{B_1}\})}{f(\bu,\bth_{B_1})}=\Lambda_1+\Lambda_2,
\end{gather*}
where
\begin{gather}
\Lambda_1=\eta\sum_{\{\bth_{A_1},\bth_{A_2}\}\vdash\bth_{A}}
\left(-\frac1\eta\right)^{\#\bth_{A_2}} f(\bu,\bth_{A_1})f(\bth_{A_1},\bth_{A_2})\nonumber\\
\hphantom{\Lambda_1=}{}\times \sum_{\substack{\{\bth_{B_1},\bth_{B_2}\}\vdash\bth_{B}\\ \#\bth_{B_1}=1}}
\frac{f(\bth_{B_2},\bth_{B_1})f(\bth_{B_1},\bth_{A_2})}{h(\theta_s,\bth_{B_1})},\label{DefL1}
\end{gather}
and
\begin{gather}
\Lambda_2=-\sum_{\{\bth_{A_1},\bth_{A_2}\}\vdash\bth_{A}}
\left(-\frac1\eta\right)^{\#\bth_{A_2}} f(\bu,\bth_{A_1})f(\bth_{A_1},\bth_{A_2})\nonumber\\
\hphantom{\Lambda_2=}{}\times
\sum_{\substack{\{\bth_{B_1},\bth_{B_2}\}\vdash\bth_{B}\\ \#\bth_{B_1}=1}}
\frac{f(\bth_{B_2},\bth_{B_1})f(\bth_{A_1},\bth_{B_1})}{h(\theta_s,\bth_{B_1})f(\bu,\bth_{B_1})}.\label{DefL2}
\end{gather}

\subsubsection{First contribution}

Consider the sum $\Lambda_1$ \eqref{DefL1}. Let
\begin{gather*}%\label{G}
G=\sum_{\substack{\{\bth_{B_1},\bth_{B_2}\}\vdash\bth_{B}\\ \#\bth_{B_1}=1}} \frac{f(\bth_{B_2},\bth_{B_1})f(\bth_{B_1},\bth_{A_2})}{h(\theta_s,\bth_{B_1})}.
\end{gather*}
Consider an auxiliary contour integral
\begin{gather*}%\label{G-int}
J=\frac1{2\pi{\rm i}c}\oint_{|z|=R\to\infty}\frac{f(\bth_{B},z)}{h(\theta_s,z)}f(z,\bth_{A_2}).
\end{gather*}
Taking the residue at infinity we obtain $J=-1$. On the other hand, this integral is equal to the sum of residues within the contour. The sum in the points $z=\theta_k\in\bth_{B}$ gives $-G$. One more series of poles occurs at $z=\theta_k\in\bth_{A_2}$. Thus, we find
\begin{gather}\label{Jr}
G=1+\sum_{\substack{\{\bth_{A_3},\bth_{A'_2}\}\vdash\bth_{A_2}\\ \#\bth_{A_3}=1}} \frac{f(\bth_{B},\bth_{A_3})}{h(\theta_s,\bth_{A_3})}f(\bth_{A_3},\bth_{A'_2}).
\end{gather}
Substituting this into \eqref{DefL1} we arrive at
\begin{gather}\label{L1-res}
\Lambda_1=\eta K_{N,n_A}^{(1/\eta)}(\bu|\bth_{A})-\tilde \Lambda_1,
\end{gather}
where
\begin{gather}
\tilde \Lambda_1=\sum_{\substack{\{\bth_{A_1},\bth_{A_2},\bth_{A_3}\}\vdash\bth_{A} \\ \#\bth_{A_3}=1}}
\left(-\frac1\eta\right)^{\#\bth_{A_2}} \frac{f(\bu,\bth_{A_1})}{h(\theta_s,\bth_{A_3})}f(\bth_{A_1},\bth_{A_2})f(\bth_{A_1},\bth_{A_3})\nonumber\\
\hphantom{\tilde \Lambda_1=}{}\times f(\bth_{A_3},\bth_{A_2})f(\bth_{B},\bth_{A_3}).\label{tL1}
\end{gather}
Here, when substituting \eqref{Jr} into \eqref{DefL1} we first set $\bth_{A_2}=\{\bth_{A_3},\bth_{A'_2}\}$ and then
relabeled $\bth_{A'_2}\to\bth_{A_2}$.

\subsubsection{Second contribution}

In order to compute the contribution $\Lambda_2$ we should transform the sum over partitions of $\bth_B$ in~\eqref{DefL2}. Let $\{\bth_B,\bth_{A_1}\}=\bth_C$. Then
\begin{gather}
\sum_{\substack{\{\bth_{B_1},\bth_{B_2}\}\vdash\bth_{B}\\ \#\bth_{B_1}=1}}
\frac{f(\bth_{B_2},\bth_{B_1})f(\bth_{A_1},\bth_{B_1})}{h(\theta_s,\bth_{B_1})f(\bu,\bth_{B_1})}=
\sum_{\substack{\{\bth_{C_1},\bth_{C_2}\}\vdash\bth_{C}\\ \#\bth_{C_1}=1}}
\frac{f(\bth_{C_2},\bth_{C_1})}{h(\theta_s,\bth_{C_1})f(\bu,\bth_{C_1})}\nonumber\\
\hphantom{\sum_{\substack{\{\bth_{B_1},\bth_{B_2}\}\vdash\bth_{B}\\ \#\bth_{B_1}=1}}
\frac{f(\bth_{B_2},\bth_{B_1})f(\bth_{A_1},\bth_{B_1})}{h(\theta_s,\bth_{B_1})f(\bu,\bth_{B_1})}=}{}
-\sum_{\substack{\{\bth_{A'_1},\bth_{A_3}\}\vdash\bth_{A_1}\\ \#\bth_{A_3}=1}}
\frac{f(\bth_{B},\bth_{A_3})f(\bth_{A'_1},\bth_{A_3})}{h(\theta_s,\bth_{A_3})f(\bu,\bth_{A_3})}.\label{sumparB}
\end{gather}
In other words, instead of taking the sum over the subset $\bth_B$, we take the sum over the subset $\bth_C=\{\bth_B,\bth_{A_1}\}$ and then subtract the sum over the subset $\bth_{A_1}$. In this case we can compute the sum over the subset $\bth_C=\{\bth_B,\bth_{A_1}\}$ via Bethe equations~\eqref{BEj}.

However, we first consider contribution coming from the second term in the r.h.s.\ of~\eqref{sumparB}. We have
\begin{gather*}
\Lambda_2^{(2)}=\sum_{\{\bth_{A_1},\bth_{A_2}\}\vdash\bth_{A}} \left(-\frac1\eta\right)^{\#\bth_{A_2}} f(\bu,\bth_{A_1})f(\bth_{A_1},\bth_{A_2})
\sum_{\substack{\{\bth_{A'_1},\bth_{A_3}\}\vdash\bth_{A_1}\\ \#\bth_{A_3}=1}}
\frac{f(\bth_{B},\bth_{A_3})f(\bth_{A'_1},\bth_{A_3})}{h(\theta_s,\bth_{A_3})f(\bu,\bth_{A_3})}.%\label{L22}
\end{gather*}
Substituting here $\bth_{A_1}=\{\bth_{A'_1},\bth_{A_3}\}$ and relabeling $\bth_{A'_1}\to \bth_{A_1}$ we find
\begin{gather}
\Lambda_2^{(2)}=\sum_{\substack{\{\bth_{A_1},\bth_{A_2},\bth_{A_3}\}\vdash\bth_{A}\\ \#\bth_{A_3}=1}}
\left(-\frac1\eta\right)^{\#\bth_{A_2}} \frac{f(\bu,\bth_{A_1})}{h(\theta_s,\bth_{A_3})}f(\bth_{A_1},\bth_{A_2})f(\bth_{A_3},\bth_{A_2})\nonumber\\
\hphantom{\Lambda_2^{(2)}=}{}\times f(\bth_{A_1},\bth_{A_3})f(\bth_{B},\bth_{A_3}).\label{L22-1}
\end{gather}
We see that this sum is equal to the term $\tilde \Lambda_1$ \eqref{tL1}. Thus, common contribution of these terms vanishes.

Consider now contribution coming from the first term in the r.h.s.\ of~\eqref{sumparB}. We have
\begin{gather}
\Lambda_2^{(1)}=-\sum_{\{\bth_{A_1},\bth_{A_2}\}\vdash\bth_{A}}
\left(-\frac1\eta\right)^{\#\bth_{A_2}} f(\bu,\bth_{A_1})f(\bth_{A_1},\bth_{A_2})\nonumber\\
\hphantom{\Lambda_2^{(1)}=}{}\times
\sum_{\substack{\{\bth_{C_1},\bth_{C_2}\}\vdash\bth_{C}\\ \#\bth_{C_1}=1}}
\frac{f(\bth_{C_2},\bth_{C_1})}{h(\theta_s,\bth_{C_1})f(\bu,\bth_{C_1})}.\label{L21}
\end{gather}
Using Bethe equations \eqref{BEj} we find
\begin{gather*}
\sum_{\substack{\{\bth_{C_1},\bth_{C_2}\}\vdash\bth_{C}\\ \#\bth_{C_1}=1}}
\frac{f(\bth_{C_2},\bth_{C_1})}{h(\theta_s,\bth_{C_1})f(\bu,\bth_{C_1})}
=\frac{\tvk+\vk}{\vk-\rho_2}-\frac1\alpha f(\bu,\theta_s)f(\theta_s,\bth_{A_2})\nonumber\\
\hphantom{\sum_{\substack{\{\bth_{C_1},\bth_{C_2}\}\vdash\bth_{C}\\ \#\bth_{C_1}=1}}
\frac{f(\bth_{C_2},\bth_{C_1})}{h(\theta_s,\bth_{C_1})f(\bu,\bth_{C_1})}=}{}
+\frac1\alpha \sum_{\substack{\{\bth_{A_3},\bth_{A'_2}\}\vdash\bth_{A_2}\\ \#\bth_{A_3}=1}}
g(\theta_s,\bth_{A_3})f(\bu,\bth_{A_3})f(\bth_{A_3},\bth_{A'_2}).%\label{BEC}
\end{gather*}
Substituting this into \eqref{L21} we arrive at
\begin{gather}\label{L21-1}
\Lambda_2^{(1)}=-\frac{\tvk+\vk}{\vk-\rho_2}K_{N,n_A}^{(1/\eta)}(\bu|\bth_A)+ \frac1\alpha(M_1-M_2),
\end{gather}
where
\begin{gather*}%\label{M1}
M_1=\sum_{\substack{\{\bth_{A_1},\bth_{A_2}\}\vdash\{\bth_{A},\theta_s\}\\ \theta_s \in\bth_{A_1}}}
\left(-\frac1\eta\right)^{\#\bth_{A_2}} f(\bu,\bth_{A_1})f(\bth_{A_1},\bth_{A_2}),
\end{gather*}
and
\begin{gather}
M_2=\sum_{\{\bth_{A_1},\bth_{A_2}\}\vdash\bth_{A}}
\left(-\frac1\eta\right)^{\#\bth_{A_2}} f(\bu,\bth_{A_1})f(\bth_{A_1},\bth_{A_2})\nonumber\\
\hphantom{M_2=}{}\times
\sum_{\substack{\{\bth_{A_3},\bth_{A'_2}\}\vdash\bth_{A_2}\\ \#\bth_{A_3}=1}}
g(\theta_s,\bth_{A_3})f(\bu,\bth_{A_3})f(\bth_{A_3},\bth_{A'_2}).\label{M2}
\end{gather}

Presenting $M_1$ as
\begin{gather*}
M_1=\sum_{\{\bth_{A_1},\bth_{A_2}\}\vdash\{\bth_{A},\theta_s\}}\left(-\frac1\eta\right)^{\#\bth_{A_2}} f(\bu,\bth_{A_1})f(\bth_{A_1},\bth_{A_2})\nonumber\\
\hphantom{M_1=}{}-\sum_{\substack{\{\bth_{A_1},\bth_{A_2}\}\vdash\{\bth_{A},\theta_s\}\\ \theta_s \in\bth_{A_2}}}
\left(-\frac1\eta\right)^{\#\bth_{A_2}} f(\bu,\bth_{A_1})f(\bth_{A_1},\bth_{A_2}),%\label{M11}
\end{gather*}
we obtain
\begin{gather*}%\label{M12}
M_1=K_{N,n_A+1}^{(1/\eta)}(\bu|\{\bth_A,\theta_s\})
-\sum_{\substack{\{\bth_{A_1},\bth_{A_2}\}\vdash\{\bth_{A},\theta_s\}\\ \theta_s \in\bth_{A_2}}}
\left(-\frac1\eta\right)^{\#\bth_{A_2}} f(\bu,\bth_{A_1})f(\bth_{A_1},\bth_{A_2}).
\end{gather*}
Setting here $\bth_{A_2}=\{\bth_{A'_2},\theta_s\}$ and relabeling $\bth_{A'_2}\to \bth_{A_2}$ we arrive at
\begin{gather*}
M_1=K_{N,n_A+1}^{(1/\eta)}(\bu|\{\bth_A,\theta_s\}) -\sum_{\{\bth_{A_1},\bth_{A_2}\}\vdash\bth_{A}}
\left(-\frac1\eta\right)^{\#\bth_{A_2}+1} f(\bu,\bth_{A_1})f(\bth_{A_1},\bth_{A_2})f(\bth_{A_1},\theta_s).%\label{M13}
\end{gather*}

Consider now contribution $M_2$. Setting $\bth_{A_2}=\{\bth_{A'_2},\bth_{A_3}\}$ in \eqref{M2} and relabeling $\bth_{A'_2}\to \bth_{A_2}$ we arrive at
\begin{gather}
M_2= \sum_{\substack{ \{\bth_{A_1},\bth_{A_2},\bth_{A_3}\}\vdash\bth_{A}\\ \#\bth_{A_3}=1} }
\left(-\frac1\eta\right)^{\#\bth_{A_2+1}} f(\bu,\bth_{A_1})f(\bth_{A_1},\bth_{A_2})f(\bth_{A_1},\bth_{A_3})
g(\theta_s,\bth_{A_3})\nonumber\\
\hphantom{M_2=}{}\times f(\bu,\bth_{A_3})f(\bth_{A_3},\bth_{A_2}).\label{M21}
\end{gather}
Let $\{\bth_{A_1},\bth_{A_3}\}=\bth_{A_0}$. Then \eqref{M21} takes the form
\begin{gather}
M_2=\sum_{\{\bth_{A_0},\bth_{A_2}\} \vdash\bth_{A} }
\left(-\frac1\eta\right)^{\#\bth_{A_2+1}} f(\bu,\bth_{A_0})f(\bth_{A_0},\bth_{A_2})\nonumber\\
\hphantom{M_2=}{}\times
\sum_{\substack{ \{\bth_{A_1},\bth_{A_3}\}\vdash\bth_{A_0}\\ \#\bth_{A_3}=1} }
g(\theta_s,\bth_{A_3})f(\bth_{A_1},\bth_{A_3}).\label{M22}
\end{gather}
The sum over partitions $\{\bth_{A_1},\bth_{A_3}\}\vdash\bth_{A_0}$ can be easily computed by the contour integral
\begin{gather*}%\label{sumA0}
\sum_{\substack{ \{\bth_{A_1},\bth_{A_3}\}\vdash\bth_{A_0}\\ \#\bth_{A_3}=1} }
g(\theta_s,\bth_{A_3})f(\bth_{A_1},\bth_{A_3})=1-f(\bth_{A_0},\theta_s).
\end{gather*}
Substituting this into \eqref{M22} we obtain
\begin{gather*}
M_2=\sum_{\{\bth_{A_0},\bth_{A_2}\}\vdash\bth_{A} }
\left(-\frac1\eta\right)^{\#\bth_{A_2+1}} f(\bu,\bth_{A_0})f(\bth_{A_0},\bth_{A_2})
\bigl(1-f(\bth_{A_0},\theta_s)\bigr)\nonumber\\
\hphantom{M_2}{} =-\frac1\eta K_{N,n_A}^{(1/\eta)}(\bu|\bth_A)-\sum_{\{\bth_{A_1},\bth_{A_2}\}\vdash\bth_{A} }
\left(-\frac1\eta\right)^{\#\bth_{A_2+1}} f(\bu,\bth_{A_1})f(\bth_{A_1},\bth_{A_2})f(\bth_{A_1},\theta_s).%\label{M23}
\end{gather*}
Thus,
\begin{gather*}%\label{M1M2}
M_1-M_2=K_{N,n_A+1}^{(1/\eta)}(\bu|\{\bth_A,\theta_s\})+\frac1\eta K_{N,n_A}^{(1/\eta)}(\bu|\bth_A),
\end{gather*}
and due to \eqref{L21-1} we arrive at
\begin{gather*}%\label{L21-2new}
\Lambda_2^{(1)}=\left(\frac1{\eta\alpha}-\frac{\tvk+\vk}{\vk-\rho_2}\right)K_{N,n_A}^{(1/\eta)}(\bu|\bth_A)+
\frac1\alpha K_{N,n_A+1}^{(1/\eta)}(\bu|\{\bth_A,\theta_s\}),
\end{gather*}
Finally, taking into account \eqref{L1-res}, \eqref{tL1}, and \eqref{L22-1} we obtain after simple algebra
\begin{gather*}%\label{L1andL2}
\Lambda_1+\Lambda_2=\frac\beta{\beta\mu+\mu-1}\Bigl(\eta K_{N,n_A+1}^{(1/\eta)}(\bu|\{\bth_A,\theta_s\})-\xi K_{N,n_A}^{(1/\eta)}(\bu|\bth_A)\Bigr).
\end{gather*}
Thus, Lemma~\ref{summin} is proved.

\subsection*{Acknowledgements}
The work was performed at the Steklov Mathematical Institute of Russian Academy of Sciences, Moscow. This work is supported by the Russian Science Foundation under grant 19-11-00062.

\pdfbookmark[1]{References}{ref}
\LastPageEnding

\end{document}